\documentclass[preprint]{imsart}

\RequirePackage{natbib}
\RequirePackage{graphicx}
\RequirePackage[final]{pdfpages}
\RequirePackage{xcolor}		
\RequirePackage{array}
\RequirePackage{booktabs}		
\RequirePackage{color}
\usepackage{float}
\usepackage[normalem]{ulem}

\RequirePackage{hyperref}

\RequirePackage[OT1]{fontenc}

\RequirePackage{amsthm,amsmath,amssymb,dsfont}
\RequirePackage{bm,bbm,mathtools}
\RequirePackage{rotating}
\RequirePackage{lscape}
\RequirePackage{here}

\startlocaldefs
\numberwithin{equation}{section}
\theoremstyle{plain}
\newtheorem{Thrm}{Theorem}[section]
\newtheorem{Lem}[Thrm]{Lemma}
\newtheorem{Cor}[Thrm]{Corollary}
\newtheorem{Prop}[Thrm]{Proposition}
\newtheorem{Ex}{Example}
\theoremstyle{remark}
\newtheorem{Rem}[Thrm]{Remark}
\newtheorem{Def}[Thrm]{Definition}

\newcommand{\N}{\mathbb{N}}

\newcommand{\R}{\mathbb{R}}

\newcommand{\E}{\mathbb{E}}
\newcommand{\PP}{\mathbb{P}}

\newcommand{\NN}{\mathcal{N}}

\newcommand{\fM}{\mathfrak{M}}

\newcommand{\Cov}{\mathrm{Cov}}
\newcommand{\Var}{\mathrm{Var}}

\newcommand{\Id}{\mathrm{Id}}
\newcommand{\dd}{\mathrm{d}}
\newcommand{\trace}{\mathrm{trace}}
\newcommand{\Exp}{\mathrm{Exp}}
\newcommand{\Log}{\mathrm{Log}}
\newcommand{\Ave}{\mathrm{Ave}}

\newcommand{\ls}{\langle}
\newcommand{\rs}{\rangle}
\DeclarePairedDelimiter\ceil{\lceil}{\rceil}
\DeclarePairedDelimiter\floor{\lfloor}{\rfloor}
\DeclareMathOperator*{\argmin}{arg\,min}

\newcommand{\wt}{\widetilde}
\newcommand{\wh}{\widehat}
\newcommand{\ol}{\overline}
\newcommand{\ul}{\underline}
\newcommand{\1}{\mathds{1}}
\renewcommand{\epsilon}{\varepsilon}
\renewcommand{\phi}{\varphi}
\newcommand*\diff{\mathop{}\!\mathrm{d}}
\newcommand{\p}{\mathbb{P}}

\newcommand{\dn}{\diagdown}

\newcommand{\dr}{\mathrm{d}}

\endlocaldefs

\allowdisplaybreaks

\begin{document}

\setcounter{page}{1}
\begin{frontmatter}
\title{Statistical inference for intrinsic wavelet estimators of SPD matrices in a log-Euclidean manifold}
\runtitle{Inference for SPD matrices}
\begin{aug}
\author{
\fnms{Johannes} \snm{Krebs} \thanksref{a} \corref{} 
}
\hspace{-.7em}
\author{
\fnms{Daniel} \snm{Rademacher} \thanksref{b} \corref{} 
}
\hspace{-.7em}
\author{
\fnms{Rainer} \snm{von Sachs} \thanksref{c} \corref{} 
}

\address[a]{Department of Mathematics, KU Eichst{\"a}tt-Ingolstadt, 85072 Eichst{\"a}tt, Germany.
E-mail: johannes.krebs@ku.de}

\address[b]{Institute of Applied Mathematics, Heidelberg University, 69120 Heidelberg, Germany. E-mail: daniel.rademacher@uni-heidelberg.de}

\address[c]{ISBA/LIDAM,  UCLouvain, 1348 Louvain-la-Neuve, Belgium. E-mail: rainer.vonsachs@uclouvain.be}

\runauthor{Krebs, Rademacher, von Sachs}
\affiliation{Heidelberg University, TU Braunschweig, Universit{\' e} Catholique de Louvain}
\end{aug}

\date{February 12, 2022}

\begin{abstract}
In this paper we treat statistical inference for an intrinsic wavelet estimator of curves of symmetric positive definite (SPD) matrices in a log-Euclidean manifold. This estimator preserves positive-definiteness and enjoys permutation-equivariance, which is particularly relevant for covariance matrices. Our second-generation wavelet estimator is based on average-interpolation and allows the same powerful properties, including fast algorithms, known from nonparametric curve estimation with wavelets in standard Euclidean set-ups. 

The core of our work is the proposition of confidence sets for our high-level wavelet estimator in a non-Euclidean geometry. We derive asymptotic normality of this estimator, including explicit expressions of its asymptotic variance. This opens the door for constructing asymptotic confidence regions which we compare with our proposed bootstrap scheme for inference. Detailed numerical simulations confirm the appropriateness of our suggested inference schemes.
\end{abstract}
%
%

\end{frontmatter}

\section{Introduction}\label{Section_Introduction}
In this paper we derive statistical inference for an intrinsic wavelet estimator of a curve of positive definite matrices in a log-Euclidean manifold. This estimator preserves, among other interesting properties, the essential property of remaining itself positive-definite (PD). The paradigm of our approach is that this preservation of PD is possible without any pre- or postprocessing step, e.g. finding the PD estimator coming closest to the proposed wavelet estimator. Such a device would be prohibitive for a feasible analysis of its statistical properties (in particular mean-square convergence and asymptotic normality of the final estimator).

Some recent work (e.g. \cite{pennec2006intrinsic}, \cite{hinkle2014intrinsic}) on using differential-geometric tools in nonparametric estimation of matrix-valued curves opened an elegant and mathematically sound way of how to transfer powerful modern non-parametric curve estimation schemes from the classical situation only involving {\em Euclidean} distances to the world of {\em non-flat} manifolds, i.e. {\em curved} spaces. A prominent example is the space of symmetric and positive-definite (SPD) matrices which we will focus upon here. Important applications for curves with such matrix-valued entries are covariance matrices (with arguments changing over time or space, e.g.), but also diffusion tensors (\cite{zhu2009intrinsic},  \cite{dryden2009non}): Diffusion tensor imaging (DTI) can be  used in clinical applications to obtain high-resolution information of internal structures of pathological versus healthy  tissues of certain organs (e.g. hearts and brains). For each tissue voxel, there is a 3 x 3 SPD matrix to describe the shape of local diffusion. Other scientific applications of SPD matrices, finally, are numerous, such as in computer vision \cite{caseiro2012computervision}, elasticity \cite{moakher2006elasticity}, signal processing \cite{arnaudson2013sp}, medical imaging (\cite{Fillard2007imaging}), \cite{fletcher2007dti}) and neuroscience \cite{Friston2011neuroscience}. 

In \cite{chau2020intrinsic} the slightly larger class of Hermitian positive-definite (HPD) matrices has been studied, with as prominent examples, spectral density matrices of multivariate time series which are functions of frequency (or even of time and frequency for locally stationary time series with a time-varying correlation structure as in  \cite{chau2021tvspectral}). It is merely for reasons of keeping the presentation sufficiently light that we restrict ourselves, in this paper, to the class of SPD matrices.

\vspace{0.3cm}

The essential key to success for the aforementioned transfer to {\em curved} spaces hinges on the following idea. Evaluating the performance of usual nonparametric curve estimation builds on measuring the distance, and hence the speed of convergence, of an estimator to the target curve by the standard Euclidean metric. Hence it is important to replace the 
latter by a suitable metric that takes into account the structure of the underlying manifold. 

While the work by \cite{chau2020intrinsic} concentrated on the use of the {\em affine-invariant} Riemannian metric, in this paper we consider the {\em log-Euclidean} metric to be at the base of our derivations. It provides a different natural distance between two SPD matrices, while still allowing to preserve some interesting invariance properties. In particular, the log-Euclidean metric transforms the space of SPD matrices in a complete metric space, and it is unitary congruence invariant (see also Lemma \ref{Lem:Invariance} in the Supplement \ref{AppendixA}).
This implies equivariance of our constructed estimators with respect to, e.g., rotations, and has as the following important application in the use of the log-Euclidean metric for estimation of covariance matrices of a multivariate data vector (e.g. a time series): if the coordinates of the entries of the data vector are permuted, the estimator of its covariance matrix follows this permutation - a property lacking for other choices, such as, e.g., the Cholesky and the log-Cholesky metrics  (\cite{lin2019riemannian}).

Using the log-Euclidean metric, most importantly, opens the way to derive properly centered asymptotic normality of the proposed curve estimators. We will subsequently use this in order to construct correctly centered confidence regions, due to the simplified geometric structure of the log-Euclidean compared to other Riemannian manifolds. Another advantage is that estimators based on this metric do not suffer from the {\em swelling effect} - which means that the determinant of an average of two (or more) SPD matrices in this metric is guaranteed to not exceed the determinant of any of the averaged matrices. This property will guarantee that our constructed confidence regions which are based on estimators that are constructed as intrinsic averages in our space of SPD matrices, will not swell either towards or beyond the borders of our manifold. This is obviously an important advantage for constructing confidence volumes not simply based on the Euclidean metric in order to have them respect the nominal level in practice without becoming incorrectly too large.

\vspace{0.3cm}

The relevant literature from the field of differential geometry on the theory of SPD matrices is vast (e.g. \cite{arsigny2007geometric}, \cite{pennec2006intrinsic}), and in this paper we restrict ourselves to review the aspect that are necessary for our approach. Essentially the idea is to relate a non-flat or curved space -  which is supposed to contain both the matrix-valued curved to be denoised and the estimator that provides this denoising - with a classical Euclidean geometry that allows the control of distances. Mathematically speaking this is achieved by studying, locally for each element of the curved space, its {\em tangent space}. There are many possibilities to do this, most of them are based on differential-geometric tools such as the Exponential mapping - and its inverse -  between the manifold and its tangent space. Among these, we decided to choose the log-Euclidean approach, in which the considered distances only depend on the matrix-logarithm of our objects.


\vspace{0.3cm}

As another paradigm of our work we want to benefit from the powerful properties of denoising observed curves by {\em wavelet estimators} (see, e.g., \cite{antoniadis1994wavelets}). However, as our observed curves are matrix-valued, classical wavelet algorithms cannot be directly applied. 
Using the log-Euclidean metric allows us to build our work on a particular {\em second-generation wavelet} scheme (for general such schemes we refer, e.g. to \cite{jansen2005second}), similarly to the one developed in \cite{chau2020intrinsic}, as an important modification of the approach of \cite{rahman2005multiscale}. This scheme is based on {\em Average-Interpolation} (AI), which is known to have powerful properties as developed, e.g. in \cite{daubechies1991two},\cite{daubechies1992two} and \cite{daubechies1992sets}. Moreover \cite{donoho1993smooth} showed that wavelet estimators based on AI carry the same desirable properties as first-generation wavelet schemes. Those allow us to control not only rates of convergence of our intrinsic wavelet estimator but go beyond this in the study of its (asymptotic) variance in our newly derived central limit theorem. By construction and, again, use of the log-Euclidean metric which avoids any swelling effects, this enables us to construct appropriate confidence regions for points on the curve with values in the space of SPD matrices.

\vspace{0.3cm}

With our contribution we complete an important gap on statistical inference for wavelet-based curve estimators on non-standard geometries. Related statistical work exists in the context of local polynomial estimation (\cite{yuan2012local}, wavelet estimation (\cite{rahman2005multiscale}), and more generally for regression on Riemannian manifolds (\cite{boumal2011discreteA}, \cite{hinkle2014intrinsic}), the latter references not providing any convergence rates of a statistical error measure (such as the mean-squared error). However, in this work we focus on providing confidence regions, both based on {\em asymptotic normality} and, for comparison and in its own interest, based on a {\em bootstrap scheme}. Beyond our detailed theoretical investigations we convince ourselves by extended simulation studies about the appropriateness of our suggested inference methods.

\vspace{0.3cm}

The rest of the paper is organized as follows. The notation and abbreviations are given in Section~\ref{Section_Notation}. We lay out conceptual details of differential geometry and intrinsic average interpolation in Section~\ref{Section_AI}. The results for wavelet estimation in our nonparametric regression setup of denoising SPD-matrix valued curves are given in Section~\ref{Section_WaveletRegression}, i.e. reviewing rates of mean-square convergence, and foremost, deriving asymptotic normality including an explicit calculation of the asymptotic variance of our estimators. Section~\ref{Section_ConfidenceSets} presents the two different constructions of our confidence sets which live in the considered non-flat space of SPD-matrix valued curves. Here we base ourselves, on the one hand, on the shown asymptotic normality, and on the other hand, on a newly proposed (wild) bootstrap scheme, for which we show its theoretical validity. Finally we provide numerical simulation studies in Section~\ref{Sec_Applications}, which show the satisfactory empirical coverage and a comparison between bootstrap confidence regions and those based on asymptotic normality. A short conclusion Section~\ref{Sec:Conclusions} sketches some possible extensions for future work. 

All technical details are given in a Supplement which is structured as follows: First, in Section~\ref{Section_Proofs} we give further insights on the construction of our AI-wavelet scheme of section 3, including some proofs. Then we proceed to proving our statistical results of sections ~\ref{Section_WaveletRegression} and ~\ref{Section_ConfidenceSets} in Sections~\ref{Sec: Proofs_Wav_Reg} and~\ref{Sec:Proofs_Bootstrap}, respectively. The Section~\ref{AppendixA} contains further material regarding the log-Euclidean metric.

\section{Preliminaries}\label{Section_Notation}

As outlined in the introduction this work is concerned with intrinsic curve estimation within the manifold $Sym^+(d)$ of
$d\times d$ symmetric positive-definite (SPD) matrices. 
Now there are several ways to equip this differentiable manifold with a Riemannian metric, that is a positive definite inner product $g_S$ on
the tangent space $T_S Sym^+(d)$ defined at each point $S \in Sym^+(d)$ and depending smoothly on $S$. We remark that $T_S Sym^+(d)$ can be 
identified with the ambient vector space $Sym(d)$ of symmetric matrices in the sense that both spaces are isomorphic.
The most natural approach to obtain a Riemannian metric on $Sym^+(d)$ is then to simply restrict the Frobenius inner product on $Sym(d)$, which leads 
to the so called \emph{affine invariant} metric, see \cite{bhatia2009positive} for a comprehensive treatment. 
Although geometrically very convenient, this metric imposes some 
computational difficulties when it comes to compute quantities such as Fr{\'e}chet means. This is inherently due to the non-flatness of the resulting
Riemannian manifold.\\ 
To circumvent some delicate issues that arise in inferential investigations on curved spaces, we instead consider the log-Euclidean metric, which was
introduced in \cite{arsigny2007geometric} and provides a \emph{flat geometry} on $Sym^+(d)$. A detailed exposition of the Riemannian structure induced 
by the log-Euclidean metric is provided in the Supplement, Section \ref{AppendixA}, of which we give a short summary in the following.\\

We start by collecting some notation: Let $M(d)$ denote the space of real $d\times d$ matrices and $GL(d)$ the 
subspace of invertible matrices. Then for any $A\in M(d)$, its \emph{matrix exponential} is given by 
$\exp(A) = \sum_{k=0}^{\infty} A^k/k! \in GL(d)$. In particular the restriction $\exp\colon Sym(d) \to Sym^+(d)$ is one-to-one and for any $S\in Sym^+(d)$,
there is a $\log(S) = A \in Sym(d)$ such that $\exp(A) = S$. The mapping $\log \colon Sym^+(d) \to Sym(d)$ is called \emph{(principal) matrix logarithm}.\\
Then, choosing the Frobenius inner product $\ls A,B\rs_{\Id} = \ls A,B\rs_F = \trace(A B)$, $A,B \in Sym(d)$ in \eqref{metricextension},
the derivations in Section \ref{AppendixA} lead to the following collection of geometric tools regarding the log-Euclidean metric:
\begin{table}[h!]
\centering
\begin{tabular}{||l | l ||} 
 \hline 
    $Sym^+(d)$ & manifold \\ [.5ex]
    \hline
    $T_S Sym^+(d) = S\times Sym(d) \cong Sym(d)$ & tangent space\\ [.5ex]
    \hline
     $g_S(U,V) = \trace(d_S\log(U) d_S\log(V)), \quad U,V \in T_S Sym^+(d)$ & log-Eucl. metric \\ [.5ex]
     \hline
     $d(S_1,S_2) = \|\log(S_2) - \log(S_1)\|_F, \quad S_1,S_2 \in Sym^+(d)$ & log-Eucl. distance \\ [.5ex]
     \hline
     $\gamma(t;S_1,S_2) = \exp((1-t)\log(S_1) + t\log(S_2)),\quad S_1,S_2 \in Sym^+(d)$ & geodesic \\ [.5ex] 
     \hline
     $\Exp_S(U) = \exp\big(\log(S) + d_S\log (U)\big),\quad S\in Sym^+(d),~ U \in T_S Sym^+(d)$ & exponential map\\ [.5ex]
     \hline
     $\Log_{S_1}(S_2) = d_{\log(S_1)}\exp\big( \log(S_2) - \log(S_1) \big), \quad S_1,S_2 \in Sym^+(d)$ & logarithmic map\\ [.5ex]
     \hline
     $\Gamma_{S_1}^{S_2}(U) = d_{\log(S_2)}\exp\big( d_{S_1}\log(U) \big),\quad S_1,S_2 \in Sym^+(d),~U \in T_{S_1} Sym^+(d)$ & parallel transport\\[.5ex]
     \hline
     $\Gamma_{S}^{\Id}(U) = d_S\log(U), \quad S\in Sym^+(d),~U \in T_S Sym^+(d)$ & \\[.5ex]
     \hline
 \end{tabular}
\caption{Abbreviations for the relevant differential geometric tools}
\label{table:Abbreviations2}
\end{table}

Turning to the probabilistic aspects of our work, we recall, from the Hopf-Rinow theorem (i.e. \cite{pennec2006intrinsic}), 
that $(Sym^+(d), \dr)$ is a complete separable metric space.
Denote by $\mathcal{B}(Sym^+(d))$ the corresponding Borel $\sigma$-algebra, then we call a measurable function
$X\colon (\Omega,\mathcal{A}, \p) \to (Sym^+(d), \mathcal{B}(Sym^+(d)))$ a \emph{random element} on $Sym^+(d)$ and the induced probability measure 
$\p^X := \p \circ X^{-1}$ on $Sym^+(d)$ is as usual referred to as the \emph{law or distribution of $X$}.\\
Further, let $P(Sym^+(d))$ denote the set of all probability measures on $(Sym^+(d),$ $\mathcal{B}(Sym^+(d)))$ and $P_m(Sym^+(d))$ the subset of 
probability measures which have a finite moment of order $m$ with respect to the log-Euclidean distance $d$, i.e.
\begin{align*}
	P_m(Sym^+(d)) := \left\{ \nu \in P(Sym^+(d))\colon \int_{Sym^+(d)} d(S_0,S)^m \, \dr \nu(S) < \infty \right\},
\end{align*}   
where $S_0\in Sym^+(d)$ is arbitrary. We note in passing that we will use $\lambda$ to denote the Lebesgue measure on $\R$.

\section{The intrinsic AI wavelet transform}
\label{Section_AI}

This section is largely inspired by the developments in \cite{chau2020intrinsic}, adapted to the particular situation of a log-Euclidean manifold. It addresses 
the three fundamental aspects of the present curve estimation problem. Due to the non-Euclidean set-up of SPD-matrix valued curves, we need to both develop how 
{\em average-interpolation} (AI) as a particular second generation schemes for wavelet estimation works for non-scalar valued curves, but also how this construction 
can be adapted to the log-Euclidean manifold. In order to do so, we first outline the geometric aspects of the present setting and provide the reader with some 
background on the {\em intrinsic}, i.e. the Fr{\' e}chet mean for the log-Euclidean metric and its connection to average interpolation (for classical wavelet AI we 
refer to \cite{klees2000wavelets} and to \cite{donoho1993smooth}). AI for wavelets provide multiscale representations and can successfully be implemented via 
{\em refinement schemes}, the basics of which we explain in the second part of this section. Here, in order to implement the {\em predicting step} that is inherent to 
all second generation schemes based on {\em lifting} (\cite{jansen2005second}), we rely on predicting the midpoints of the refinement scheme by {\em intrinsic  average}
interpolation. In the third part, we detail the concept of forward and backward average interpolation (which provide the fast wavelet transform), hence our multiscale
algorithm for processing observed data and their estimated local averages across different scales. Finally, in the fourth part, we discuss convergence of AI refinement schemes.

\subsection{Fr{\' e}chet means}
\label{Subsec:FrechetMean}
Let $\E[X]$ or $\mu$ denote the Fr{\' e}chet (or Karcher) mean of a random element $X$ on $Sym^+(d)$, which provides a general notion of location in metric spaces
and is in our case defined as the set
\begin{align*}
	\E[X] := \argmin_{R\in Sym^+(d)} \E[d^2(R,X)] =
	\argmin_{R\in Sym^+(d)} \int_{Sym^+(d)} d^2(R,S) \, \dr \PP^X(S).
\end{align*}

If $\PP^X \in P_2(Sym^+(d))$, then at least one such point exists. Moreover, the Fr{\' e}chet mean $\E[X]$ is unique because $Sym^+(d)$ is a geodesically complete 
and flat manifold with respect to the log-Euclidean metric. Since at every point $S\in Sym^+(d)$ the exponential map (see table \ref{table:Abbreviations2}) is defined 
on the entire tangent space $T_SSym^+(d)$ the Fr{\' e}chet mean $\E[X] \in Sym^+(d)$ is also
uniquely characterized by the condition $\E [\Log_{\mu}(X)] = 0$. The simple geometry under the log-Euclidean metric  then allows to conclude
\begin{align*}
	\E[X] = \exp(\E[\log(X)])\ ,
\end{align*}
where $\E[\log(X)]$ is just the usual Euclidean expectation of the random matrix $\log(X)$ in the vector space $Sym(d)$. 
That is, for $\PP^{\log(X)} = \PP^X \circ \exp$, we have
\begin{align*}
	\E[\log(X)] = \int_{Sym(d)} A \, \dr \PP^{\log(X)}(A) = \int_{Sym^+(d)} \log(S) \, \dr \PP^{X}(S).
\end{align*}
If $X$ has a discrete distribution $\PP^X = \sum_{i=1}^n w_i\delta_{S_i}$, where $\sum_{i} w_i = 1$, $w_i > 0$ and $S_i \in Sym^+(d)$, then 
\begin{align*}
    E[X] = \exp\left( \sum_{i=1}^n w_i \log(S_i) \right) =: \mathrm{Ave}(\{ S_i \}; \{w_i\}),
\end{align*}
see also \cite{arsigny2007geometric}, Theorem 3.13. In case $w_i = 1/n$ we may simply write $\mathrm{Ave}(\{ S_i \})$ (or $\overline{S}_n$ resp.\ $\overline{X}_n$) 
for the average $\mathrm{Ave}(\{ S_i \}; \{1/n\})$.\\

Furthermore, if we consider a smooth curve $c\colon \R \supset I \to Sym^+(d)$ as a mapping from the probability space $(I,\mathcal{B}(I), \lambda/\lambda(I))$ to 
the 1-dimensional submanifold $c(I) \subset Sym^+(d)$ we can also define the \emph{intrinsic mean $\mathrm{Ave}_I(c)$ of $c$} by
\begin{align*}
	\mathrm{Ave}_I(c) := \argmin_{R\in c(I)} \int_{c(I)} d^2(R,S) \, \dr\nu^{c}(S)
\end{align*}
with $\nu^c := \lambda/\lambda(I) \circ c^{-1}$. Again in case $\nu^c \in P_2(c(I))$, $\mathrm{Ave}_I(c)$ exists and is unique. Restricting the above considerations 
on the submanifold $c(I)$ we then obtain
\begin{align}\label{IntrinsicMeanOfCurve}
	\mathrm{Ave}_I(c) = \exp\left\{ \int_{c(I)} \log(S) \, \dr \nu^c(S) \right\} 
	= \exp\left\{ \frac{1}{\lambda(I)} \int_I \log(c(t)) \, \dr  t \right\}.
\end{align}   

\subsection{Intrinsic average-interpolation refinement scheme}
\label{Subsec:AI-RefinementScheme}
The data are assumed to stem from
a curve $c\colon [0,1] \to Sym^+(d)$, which is square integrable in that $\int_0^1 \|c(t)\|_F^2 \, \dr t < \infty$. As usual for nonparametric regression, the $n=2^J$
data observations $(c(k/2^J))_{k=0, \ldots, 2^J-1}$ result from (noisy versions of) equidistant sampling of the curve $c(t)$ on a (dyadic) resolution scale  $J$, i.e. 
for grid points $t_k=k/2^J, k=0, \ldots, 2^J-1$. As a common possibility for refinement schemes based on wavelets, those sampled data points are identified with the 
{\em scaling coefficients} $(M_{J,k})_k$ of the (wavelet) refinement scheme on the sampling (or finest) scale $J$.  That is, the \emph{input data} of the refinement
scheme are the $(M_{J,k})_k$, observed on scale $J$, and satisfy the following deterministic relation
\begin{align*}
	M_{J,k} = \mathrm{Ave}_{I_{J,k}}(c) 
	= \exp\left\{ 2^J \int_{I_{J,k}} \log(c(t)) \, \dr  t \right\}, \quad k=0,\ldots, 2^J - 1,
\end{align*} 
where $I_{J,k} = 2^{-J}[k,k+1)$ is a uniform partition of $[0,1)$. \\[5pt]  
We note that the intrinsic refinement scheme described in the following can also be adapted to a {\em non-dyadic} observation grid, see for instance \cite{chau2018}.\\ 

In order to later utilize the refinement scheme for regression analysis the requirements for the scheme are 
two folded. On the one hand, refined or predicted midpoints have to form an average of coarser scale midpoints, so the transformation is suitable for noise-removal. 
On the other hand, the success of the refinement schemes hinges on the ability to reconstruct what is to be defined an {\em intrinsic} polynomial of a given degree
without any error of approximation, akin the essential property 
that is enjoyed by wavelet schemes (of adapted order) in the classical setting. Similarly to the usual definition of a polynomial of degree $k$ to have vanishing
derivatives of all orders $\ell \geq k$, here now a smooth curve 
$c\colon \R \supset [a,b] \to M$ on a manifold $M$ with existing (covariant) derivatives of all orders is said to be a \emph{polynomial curve} or
\emph{intrinsic polynomial} of degree $k$, if all its derivatives of orders $\ell \geq k$ do vanish. For example a zero-degree polynomial is just a constant curve 
and a first-degree polynomial corresponds to 
a geodesic curve. For a more formal definition, we refer to \cite{hinkle2014intrinsic}, see also \cite{chau2020intrinsic}. Unfortunately intrinsic $n$th-degree
polynomials are in general difficult to express in closed form.
However, and this will be important for our purposes, the unique interpolation polynomial passing through $S_0,\ldots,S_n \in Sym^+(d)$ can be evaluated via an 
intrinsic version of Neville's algorithm, see in the supplement, section \ref{Section_Proofs} for details. 

Now to derive a suitable refinement scheme that fits those requirements we adapt the average-interpolation (AI) refinement scheme as introduced in 
\cite{donoho1993smooth} to the log-euclidean setup. 
That is, based on $j$-scale midpoints $(M_{j,k})_k$ the predicted or refined $(j+1)$-scale midpoints, denoted $(\wt{M}_{j+1,k})_k$, are computed as the $(j+1)$-scale
midpoints of unique intrinsic polynomials with 
$j$-scale midpoints $(M_{j,k})_k$. To elaborate, let $N = 2L+1$ for some $L\geq 0$ and fix a location $k \in \{ L,\ldots, 2^j-L \}$. The intrinsic AI-refinement scheme 
of order or degree $N$ is then formally defined as follows:\\
Let $\rho = \rho_{j,k} \colon [0,1] \to Sym^+(d)$ be the unique intrinsic polynomial (of degree $N-1$) with $j$-scale midpoints $M_{j,k-L},\ldots, M_{j,k+L}$, i.e. 
\begin{align*}
    \Ave_{I_{j,k'}}(\rho) = \exp\left\{ 2^j \int_{k'2^{-j}}^{(k'+1)2^{-j}} \log(\rho(t)) \,\dd t \right\} = M_{j,k'}, \quad k'=k-L,\ldots,k+L. 
\end{align*}
Then calculate the two predicted midpoints $\wt{M}_{j+1,2k}$ and$\wt{M}_{j+1,2k+1}$ on the next finer scale as 
\begin{align*}
    \wt{M}_{j+1,2k} = Ave_{I_{j+1,2k}}(\rho) = \exp\left\{ 2^{j+1} \int_{k2^{-j}}^{(2k+1)2^{-(j+1)}} \log(\rho(t)) \,\dd t \right\},\\
    \wt{M}_{j+1,2k+1} = Ave_{I_{j+1,2k}}(\rho) = \exp\left\{ 2^{j+1} \int_{(2k+1)2^{-(j+1)}}^{(k+1)2^{-j}} \log(\rho(t)) \,\dd t \right\}.
\end{align*}
Computing the integrals on the right hand side directly is of course infeasible, but similarly to AI-refinement on the real line, we can take advantage of the fact
that the mean (in the sense of \eqref{IntrinsicMeanOfCurve}) of an intrinsic polynomial is again an intrinsic polynomial to obtain tractable formulas for the
predicted midpoints. To that end denote the cumulative 
intrinsic mean of $\rho$ by 
\begin{align*}
    \mathrm{Mean}(\rho; x) := \Ave_{[(k-L)2^{-j}, x]}(\rho),\qquad x\in [(k-L)2^{-j}, 1].
\end{align*}
Then, for $ l=1,\ldots, 2L+1=N$, a simple calculation shows that
\begin{align*}
    \ol{M}_{j,l} := \mathrm{Mean}(\rho; (k-L+l)2^{-j}) 
    = \Ave(M_{j,k-L},\ldots, M_{j,k-L+l}).
\end{align*}
Furthermore we have 
\begin{align}\label{Eq:AI-Refinement1}
    \mathrm{Mean}(\rho; (2k+1)2^{-(j+1)}) 
    =& \exp\left\{ \frac{2L}{2L+1} \log(\ol{M}_{j,L}) + \frac{1}{2L+1} \log(\wt{M}_{j+1,2k}) \right\} \nonumber\\
    =& \gamma(\frac{1}{2L+1}; \ol{M}_{j,L}, \wt{M}_{j+1,2k} ),
\end{align}
that is $\mathrm{Mean}(\rho; (2k+1)2^{-(j+1)})$ lies on the geodesic segment connecting $\ol{M}_{j,L}$ and $\wt{M}_{j+1,2k}$. Since $\mathrm{Mean}(\rho; \cdot)$ 
is by construction again an intrinsic polynomial of degree $N-1$ it can be reconstructed by a suitable polynomial interpolation. To that end let $\pi$ denote the
intrinsic interpolation polynomial determined by the constrains 
\begin{align*}
    \pi((k-L+1)2^{-j}) = \ol{M}_{j,l} = \Ave(M_{j,k-L},\ldots, M_{j,k-L+l}),
\end{align*}
for $l=1,\ldots,N$. Then, due to the uniqueness of the interpolation polynomial, we have $\pi((2k+1)2^{-(j+1)}) =  \mathrm{Mean}(\rho; (2k+1)2^{-(j+1)})$ and
\eqref{Eq:AI-Refinement1} yields 
\begin{align}\label{predictionEven}
    \wt{M}_{j+1,2k} = \gamma(-2L; \pi((2k+1)2^{-(j+1)}), \ol{M}_{j,L}).
\end{align}
For the other predicted midpoint note that 
\begin{align*}
    \Ave(\wt{M}_{j+1,2k}, \wt{M}_{j+1,2k+1}) = M_{j,k}
\end{align*}
and therefore solving for $\wt{M}_{j+1,2k+1}$ yields 
\begin{align}\label{predictionOdd}
    \wt{M}_{j+1,2k+1}
    = \exp\left\{2\log(M_{j,k}) - \log(\widetilde{M}_{j,2k})\right\}.
\end{align}
Derivations of \eqref{Eq:AI-Refinement1} and \eqref{predictionEven} are provided in the supplement, section \ref{Section_Proofs}.\\

\textbf{Faster midpoint prediction in practice.}
The AI refinement scheme \eqref{predictionEven} and \eqref{predictionOdd} reduces in practice to elementary linear algebra, which is encoded in a vector of prediction
coefficients. We show the derivation of these coefficients for $N=1$ and $N=3$. Let $( M_{j,k-L},\ldots, M_{j,k+L} )$ be a given input on scale $j$.\\

\textit{(a) $N=1$ (i.e., $L=0$)}. The prediction scheme simply forwards the values from scale $j$ to scale $j+1$, viz.,
\begin{align*}
    \wt{M}_{j+1,2k} &= \exp\left\{ \log(\pi((2k+1)2^{-(j+1)})) \right\} = \ol{M}_{j,k} = \Ave(M_{j,k}) = M_{j,k},\\
    \wt{M}_{j+1,2k+1} &= \exp\left\{ 2\log(M_{j,k}) - \log(\wt{M}_{j,2k}) \right\} = M_{j,k}.
\end{align*}

\textit{(b) $N=3$ (i.e, $L=1$)} is the first non-trivial case. We show in Section~\ref{S:DetailsMidpointPrediction}
\begin{align}
	    \wt{M}_{j+1,2k} &= \exp\left\{ \frac{1}{8}\log(M_{j,k-1}) + \log(M_{j,k}) - \frac{1}{8}\log(M_{j,k+1}) \right\}, \label{E:FastPredictEx1}\\
    \wt{M}_{j+1,2k+1}  &= \exp\left\{ -\frac{1}{8}\log(M_{j,k-1}) + \log(M_{j,k}) + \frac{1}{8}\log(M_{j,k+1}) \right\}.\label{E:FastPredictEx2}
\end{align}
Carrying out these steps for general $L\ge 0$, i.e., $N=2L+1$, we obtain
\begin{align}\begin{split}\label{E:ScalarPrediction}
    \wt{M}_{j+1,2k} &= \Ave(\{M_{j,k-L},\ldots, M_{j,k+L}\}, \\
    &\qquad\qquad \{-c_L,\ldots,-c_1,1,c_1,\ldots,c_L\})\\
    \wt{M}_{j+1,2k+1} &= \Ave(\{M_{j,k-L},\ldots, M_{j,k+L}\},\\
    &\qquad\qquad \{c_L,\ldots,c_1,1,-c_1,\ldots,-c_L\})
    \end{split}
\end{align}
for suitable weights $(c_1,\ldots,c_L)$. In fact these weights are the same as in the average-interpolation transform in the Euclidean case (see \cite{donoho1993smooth}). 
If $L=1$, $c_1 = -1/8$ by the above derivation. If $L=2$, $(c_1,c_2) = (22, -3)/128$. If $L=3$, $(c_1,c_2,c_3)=(-201, 44, -5)/1024$.

\subsection{Intrinsic forward and backward average-interpolation wavelet transform} 
The intrinsic AI-refinement scheme developed in the previous section naturally leads to a corresponding intrinsic AI wavelet transform that can 
be used to synthesize a smooth curve in $Sym^+(d)$.
The first step in constructing a wavelet transform is to build up a redundant midpoint (or scaling coefficient) pyramid. The midpoint pyramid
algorithm aggregates information from a finest observation scale $J$ successively to coarser scales. It is at the heart of the wavelet transform.\\

\begin{figure}[H]
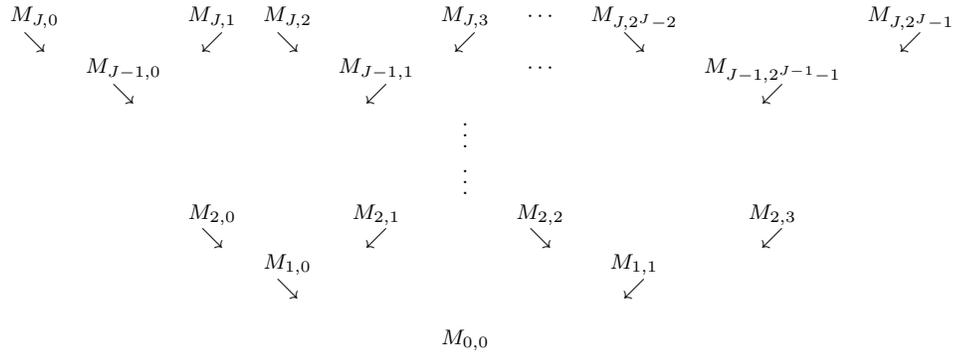

\begin{align*}
	\begin{array}{cccccccccc}
	M_{J,0}  & 			  &  M_{J,1}  & M_{J,2}  & 		     & M_{J,3}  & \cdots      & M_{J,2^J-2}  &   				 & M_{J,2^J-1} \\
	\searrow & 			  & \swarrow  & \searrow & 		     & \swarrow & 	 	             & \searrow     &   				 & \swarrow    \\
			 & M_{J-1,0}  &			  & 		 & M_{J-1,1} & 			& \cdots      &    			 & M_{J-1,2^{J-1}-1} &			   \\
			 & \searrow   & 		  &          & \swarrow  &			& 		   		  &				 & \swarrow			 & 			   \\
			 &			  &			  &			 & 			 & \vdots   &		  	      &				 &					 &			   \\
	 		 &			  &			  &			 & 			 & \vdots   &		  	      &				 &					 &			   \\
	         & 			  & M_{2,0}   &			 & M_{2,1}   &			& M_{2,2}  &	   	  & M_{2,3}		 		     &			   \\
	         &		      & \searrow  & 		 & \swarrow  & 		    & \searrow & 		  & \swarrow  	 & 					  			   \\
	         &			  &			  & M_{1,0}	 &				&	&		   & M_{1,1}  & 			 					 &			   \\
	         &			  &		      & \searrow &			 &			&		   & \swarrow &			     				 &			   \\
	         	         &			  &			  &			 &			 &   &		   &		  &				 &\\
	         &			  &			  &			 &			 & M_{0,0}  &		   &		  &				 &									 			 
	\end{array}
\end{align*}
\caption{Scheme of the midpoint pyramid algorithm. Each $M_{j,k}$ is obtained from $M_{j+1,2k}$ and $M_{j+1,2k+1}$ by the averaging procedure in \eqref{midpointPyramid},
resulting in a consecutive finer average for a fixed $M_{j,k}$ in the case that the finest scale $J$ increases to $\infty$.}
\label{Fig:MidpointPyramidAlgo}
\end{figure} 
\textbf{Midpoint pyramid.}
Consider midpoints $(M_{J,k})_{k=0}^{2^J-1}$ at the \emph{finest scale} $J$ to be given by the sampled data. At the next coarser scale $j=J-1$ set 
$M_{j,k}$ to be the
halfway point or \emph{midpoint} on the  geodesic $\gamma$ passing through $M_{j+1,2k}$ and $M_{j+1,2k+1}$, i.e. 
\begin{align}\begin{split}\label{midpointPyramid}
	M_{j,k} &:= \gamma(1/2; M_{j+1,2k},M_{j+1,2k+1}) = \Ave(M_{j+1,2k}, M_{j+1,2k+1})
\end{split}\end{align}
for $k=0,\ldots, 2^j -1$.
Continue this coarsening operation up to 
scale $j=0$ to obtain the \emph{midpoint pyramid} $((M_{j,k})_{j,k}: 0\le j \le J, 0\le k \le 2^j-1 )$, see also the scheme in Figure~\ref{Fig:MidpointPyramidAlgo}.


Now the above considerations lead to the following forward ($j\leadsto j-1$) and backward ($j-1 \leadsto j$) AI-wavelet transform.\\

\textbf{Forward wavelet transform.}
As \textit{input} take the $j$-scale midpoints $(M_{j,k})_k$. The \textit{output} is the sequence 
of $(j-1)$-scale midpoints $(M_{j-1,k})_k$ and $j$-scale \emph{wavelet coefficients} $(D_{j,k})_k$. In detail, the forward transform works as follows.
\begin{enumerate}
    \item[1.]\textit{Coarsen/predict.}\\ 
    (i) Compute the $(j-1)$-scale midpoints $(M_{j-1,k})_k$ according to the midpoint relation \eqref{midpointPyramid}.\\  
    (ii) Select a refinement order $N=2L+1$, $L\geq 0$ and generate the predicted midpoints $(\widetilde{M}_{j,k})_k$ based on $(M_{j-1,k})_k$ via 
    \eqref{predictionEven} and \eqref{predictionOdd} respectively \eqref{E:ScalarPrediction}.
    
    \item[2.]\textit{Difference.} Define the \emph{wavelet coefficients} $(D_{j,k})_{k}$ as an intrinsic difference according to 
    \begin{align}\label{waveletCoeff}
        D_{j,k} := 2^{-j/2} \Log_{\widetilde{M}_{j,2k+1}}(M_{j,2k+1}) \in T_{\widetilde{M}_{j,2k+1}}Sym^+(d).
    \end{align}
\end{enumerate}

\begin{Rem}
By construction the wavelet coefficients live in different tangent spaces. However, in order to be able to compare the wavelet coefficients (e.g., with a given uniform threshold), parallel transporting them to the same tangent space at the 
identity results in so called \emph{whitened wavelet coefficients} given by 
    \begin{align}\begin{split}\label{whitenedWaveletCoeff}
	\mathfrak{D}_{j,k} &:= \Gamma_{\widetilde{M}_{j,2k+1}}^{\Id}(D_{j,k}) \\
		&= 2^{-j/2} \big( \log(M_{j,2k+1}) - \log(\widetilde{M}_{j,2k+1}) \big) \in T_{\Id}Sym^+(d) = Sym(d).
	\end{split}
\end{align}
Of course the "length" of the vectors does not change, that is
\begin{align*}
    g_{\widetilde{M}_{j,2k+1}}( D_{j,k}, D_{j,k} )
    = 2^{-j} \| \log(M_{j,2k+1}) - \log(\widetilde{M}_{j,2k+1}) \|_F^2
    = \| \mathfrak{D}_{j,k} \|_F^2.
\end{align*}
\end{Rem}


\textbf{Backward wavelet transform.} As \textit{input} take the $(j-1)$-scale midpoints $(M_{j-1,k})_{k}$ and $j$-scale wavelet coefficients $(D_{j,k})_{k}$. 
The \textit{output} is the sequence of $j$-scale midpoints $(M_{j,k})_{k}$. In detail, the backward transform is performed as follows.
\begin{enumerate}
    \item [1.]\textit{Refine/predict.}
    (i) Given the refinement order $N=2L+1$, $L\geq 0$, generate the odd predicted midpoints 
    $(\widetilde{M}_{j,2k+1})_{k}$ for $k\in \{0,\ldots,2^{j-1}-1\}$ based on $(M_{j-1,k})_{k}$ via (\ref{predictionOdd}).\\ 
    (ii) Compute the odd $j$-scale midpoints $(M_{j,2k+1})_{k}$ for $k\in\{ 0,\ldots,2^{j-1}-1\}$ by reversing (\ref{waveletCoeff}), i.e.,
    \begin{align}\label{ReconstructionOdd}
        M_{j,2k+1} = \Exp_{\widetilde{M}_{j,2k+1}}(2^{j/2}D_{j,k}).
    \end{align}
    \item[2.]\textit{Complete.} Compute the even $j$-scale midpoints $(M_{j,2k})_{k}$ through the midpoint relation (\ref{midpointPyramid}) for
    $k\in\{1,\ldots,2^{j-1}-1\}$, i.e.,
    \begin{align}\label{ReconstructionEven}
        M_{j,2k}=\exp\left(2\log(M_{j-1,k}) - \log(M_{j,2k+1})\right)
    \end{align}
\end{enumerate}

Thus, given the coarsest midpoint $M_{0,0}$ and the wavelet coefficient pyramid $(D_{j,k})_{j,k}$ for $j=0,\ldots,J$ and $k=0,\ldots,2^{j-1}-1$, the 
original input sequence $(M_{J,k})_{k}$ can be retrieved by repeating the reconstruction procedure \eqref{ReconstructionOdd} and 
\eqref{ReconstructionEven} up to scale $J$.

\subsection{Convergence of AI refinement schemes}\label{Subsec:fast_midpoint_pred}

The estimation procedure which we introduce in the subsequent section \ref{Section_WaveletRegression} will require multiple consecutive applications of the AI
refinement scheme when lifting the wavelet estimator from a coarse scale $J_0$ to a finer scale $J$. To consider such iterative applications it turns out
to be convenient to express the transformations \eqref{E:ScalarPrediction} via suitable transition matrices. In particular this allows us to study the 
asymptotic behaviour of our wavelet estimator.\\
To elaborate, let us define for a generic $m\times n$ matrix
$A=(a_{i,j})_{i,j}$, the blocked $md\times nd$ matrix $A_{\dn d}$ by
$$
	A_{\dn d} = ( a_{i,j} I_d )_{i,j},
$$
where $I_d\in\R^{d\times d}$ is the identity matrix. This means $A_{\dn d}$ consists of block matrices each of these matrices being the $a_{i,j}$ multiple of the
identity $I_d$. Using the weights $c_i$, $i\in \{1,\ldots,L\}$, we define the square matrices $E_{N}, O_N \in \R^{(2N-1)\times(2N-1)}$, see equations \eqref{E:EN} 
and \eqref{E:ON} at the end of this section. (Notice that $E_N$ and $O_N$ are similar as $O_N = Z E_N Z$, where $Z=Z^{-1}$ with $Z = [e_{2N-1},\ldots,e_1]$ and $e_i$ being the $i$th standard 
basis (column) vector in $\R^{(2N-1)}$.) Then a careful observation shows that the AI refinement \eqref{E:ScalarPrediction} can be described with $(E_N)_{\dn d}$ resp.
$(O_N)_{\dn d}$ by
\begin{align}\label{EvenTransition}
    \begin{bmatrix}
    \log \wt{M}_{j+1,2k-2L}\\
    \vdots\\
    \log \wt{M}_{j+1,2k+2L}
    \end{bmatrix}
    =
   ( E_N)_{\dn d}
     \begin{bmatrix}
    \log M_{j,k-2L}\\
    \vdots\\
    \log M_{j,k+2L}
    \end{bmatrix},\\
\label{OddTransition}
    \begin{bmatrix}
    \log \wt{M}_{j+1,2k-2L+1}\\
    \vdots\\
    \log \wt{M}_{j+1,2k+2L+1}
    \end{bmatrix}
    =
   ( O_N)_{\dn d}
    \begin{bmatrix}
    \log M_{j,k-2L}\\
    \vdots\\
    \log M_{j,k+2L}
    \end{bmatrix}.   
\end{align}
Consequently, we have
\begin{Lem}\label{L:j-stepPrediction}
For $m\geq 1$ and $p=0,\ldots, 2^m-1$ we have
\begin{align*}
    \begin{bmatrix}
    \log \wt{M}_{j+m,2^mk-2L+p}\\
    \vdots\\
    \log \wt{M}_{j+m,2^mk+2L+p}
    \end{bmatrix}
    = 
    X_m X_{m-1} \ldots X_1
    \begin{bmatrix}
    \log M_{j,k-2L}\\
    \vdots\\
    \log M_{j,k+2L}
    \end{bmatrix}
\end{align*}
with $X_i \in \{ (E_N)_{\dn d}, (O_N)_{\dn d} \}$, $i=1,\ldots,m$.
\end{Lem}
So, it is of crucial interest whether the sequences $X_m\cdot X_{m-1} \cdot\ldots \cdot X_1$ converge as $m$ increases to infinity. 
The answer is affirmative and follows from results in \cite{daubechies1991two, daubechies1992two} as well as \cite{donoho1993smooth}.\\
To see this, we have to first define 
\begin{align*}
    (\wh{c}_{-2L},\ldots, \wh{c}_{2L+1}) 
    := ( c_L, -c_L,\ldots, c_1, -c_1, 1, 1, -c_1, c_1, \ldots, -c_L, c_L )
\end{align*}
The unique non-trivial solution of the two-scale difference equation 
\begin{align}\label{E:TwoScaleDiffEquation}
    f(x) = \sum_{k} \wh{c}_{k} f(2x-k), \qquad f\in L^1
\end{align}
is then given by the limit (that is an application ad infinitum) of the classical AI refinement scheme on the real line to the Kronecker sequence 
$(\delta_{0,k})_k$, see \cite{donoho1993smooth} Theorem 2.1 and \cite{daubechies1992two} Theorem 2.2 for a detailed explanation. Let $\phi_L$ denote 
this \emph{fundamental solution} of \eqref{E:TwoScaleDiffEquation}. For further properties of $\phi_L$, in particular its regularity, we refer to the 
aforementioned papers. Now the connection between $\phi_L$ and the limit of products of the form $X_m\cdot X_{m-1} \cdot\ldots \cdot X_1$, is as follows:

For any $x\in[0,1]$ we have the binary expansion 
\begin{align*}
    x = \sum_{j=1}^{\infty} d_j(x) 2^{-j}, \qquad d_j(x) \in \{0,1\}. 
\end{align*}
Adapting the notation in \cite{daubechies1992two} let us define 
\begin{align*}
    T_0 = Z E_N^T Z, \qquad T_1 = Z O_N^T Z.
\end{align*}
Then as a consequence of Theorem 2.2 in \cite{daubechies1992two} we have that all infinite products of $T_0$ and $T_1$ converge to the limit 
\begin{align*}
    \lim_{m \to \infty} T_{d_1(x)}  T_{d_2(x)} \ldots T_{d_m(x)}
    = 
    \begin{bmatrix}
    \phi_L(x-2L)  & \cdots & \phi_L(x-2L) \\
    \vdots &                         & \vdots \\
    \phi_L(x+2L)  & \cdots & \phi_L(x+2L)
    \end{bmatrix},
\end{align*}
where $(d_j(x))_j$ are the coefficients of the binary expansion of $x\in[0,1]$. Now let
\begin{align*}
    X_{d_j(x)} = \1\{ d_j(x) =0 \} (E_N)_{\backslash d} + \1\{ d_j(x) =1 \} (O_N)_{\backslash d}
\end{align*}
then 
\begin{align*}
    X_{d_1(x)}^T  X_{d_2(x)}^T  \ldots  X_{d_m(x)}^T = Z_{\backslash d}   T_{d_1(x)}  T_{d_2(x)} \ldots  T_{d_m(x)}  Z_{\backslash d}
\end{align*}
and therefore the above convergence implies for each $x\in[0,1]$
\begin{align}\label{E:LimitMatrix}
    \lim_{m \to \infty} X_{d_m(x)} \ldots   X_{d_2(x)} X_{d_1(x)}
    =     
    \begin{bmatrix}
    \phi_L(x-2L) &  \cdots & \phi_L(x+2L) \\
    \vdots &                          & \vdots \\
    \phi_L(x-2L) & \cdots & \phi_L(x+2L)
    \end{bmatrix}_{\backslash d}
    =: \Phi_L(x)_{\backslash d}.
\end{align}

\begin{landscape}
{\small
\begin{align}
    E_N
    &:=
	\left[
	\begin{array}{ccccccccccccccccc}
    -c_L & -c_{L-1} & \cdots & -c_1   & 1      & c_1      & \cdots & c_{L-1}  & c_L      & 0    &        &          &        &        &         &  \\
    c_L  & c_{L-1}  & \cdots & c_1    & 1      & -c_1     & \cdots & -c_{L-1} & -c_L     & 0        &        &          &         &        &         & \\
    0    & -c_L     & \cdots & -c_2   & -c_1   & 1        & \cdots & c_{L-2}  & c_{L-1}  & c_L      &        &          &         &        &         & \\
    0    & c_L      & \cdots & c_2    & c_1    & 1        & \cdots & -c_{L-2} & -c_{L-1} & -c_L     &        &          &         &        &         & \\
         &          &        & \vdots & \vdots & \vdots   &        & \vdots   & \vdots   & \vdots   &        &          &         &        &         & \\
         &          &        & 0      & -c_L   & -c_{L-1} & \cdots & -c_1     & 1        & c_1      & \cdots & c_{L-1}  & c_L     & 0      &         & \\
         &          &        & 0      &  c_L   &  c_{L-1} & \cdots &  c_1     & 1        & -c_1     & \cdots & -c_{L-1} & -c_L    & 0      &         & \\
         &          &        &        &        &          &        &  \vdots  & \vdots   & \vdots   &        & \vdots   & \vdots  &        &         & \\
         &          &        &        &        &          & 0      & -c_L     & -c_{L-1} & -c_{L-2} & \cdots & 1        & c_1     & \cdots & c_L     & 0\\
         &          &        &        &        &          & 0      & c_L      & c_{L-1}  & c_{L-2}  & \cdots & 1        & -c_1    & \cdots & -c_L    & 0\\
         &          &        &        &        &          &        & 0        & -c_L     & -c_{L-1} & \cdots & -c_1     & 1       & \cdots & c_{L-1} & c_L 
    \end{array}
    \right]
 \label{E:EN} \tag{$E_N$}\\
 \nonumber \\
  \nonumber \\
   \nonumber \\
    O_N
    &:=
    \left[
    \begin{array}{ccccccccccccccccc}
    c_L & c_{L-1} & \cdots & c_1    & 1      & -c_1     & \cdots & -c_{L-1}  & -c_L      & 0        &        &          &         &        &         &          & \\
    0   & -c_L    & \cdots & -c_2   & -c_1   & 1        & \cdots & c_{L-2}   & c_{L-1}   & c_L      &        &          &         &        &         &          & \\
    0   & c_L     & \cdots & c_2    & c_1    & 1        & \cdots & -c_{L-2}  & -c_{L-1}  & -c_L     &        &          &         &        &         &          & \\
        &         &        &        & \vdots & \vdots   &        & \vdots    & \vdots    & \vdots   &        &          &         &        &         &          & \\
        &         &        & 0      & c_L    & c_{L-1}  & \cdots & c_1       & 1         & -c_1     & \cdots & -c_{L-1} & -c_L    & 0      &         &          & \\
        &         &        & 0      & 0      & -c_L     & \cdots & -c_2      & -c_1      & 1        & \cdots & c_{L-2}  & c_{L-1} & c_L    &         &          & \\
        &         &        &        &        &          &        &  \vdots   & \vdots    & \vdots   &        & \vdots   & \vdots  & \vdots &         &          & \\
        &         &        &        &        &          &        & -c_L      & -c_{L-1}  & -c_{L-2} & \cdots & c_1      & c_2     & c_3    & \cdots  & c_L      & 0 \\
        &         &        &        &        &          &        & c_L       & c_{L-1}   & c_{L-2}  & \cdots & -c_1     & -c_2    & -c_3   & \cdots  & -c_L     & 0 \\
        &         &        &        &        &          &        & 0         & -c_L      & -c_{L-1} & \cdots & 1        & c_1     & c_2    & \cdots  & c_{L-1}  & c_L \\
        &         &        &        &        &          &        & 0         & c_L       & c_{L-1}  & \cdots & 1        & -c_1    & -c_2   & \cdots  & -c_{L-1} & -c_L
    \end{array}
    \right]
\label{E:ON} \tag{$O_N$}
\end{align}
}
\end{landscape}

\section{Wavelet regression for smooth curves in \texorpdfstring{$Sym^+(d)$}{Sym+(d)} }
\label{Section_WaveletRegression}

Having introduced in the previous section the AI-wavelet scheme for SPD-matrix valued data, we can now describe our procedure of how to estimate smooth curves of symmetric PD matrices, observed in the presence of noise. For this we begin by introducing the data generating process given a certain curve $c\colon [0,1] \to Sym^+(d)$. Following the usual paradigm of (scalar) wavelet estimation, we set 
$$
	M_{J,k} = \mathrm{Ave}_{I_{J,k}}(c)\ , \; \; k\in\{0,\ldots,n-1\} ,
$$
for $n=2^J$. We study an independent sample $M_{J,0,n},\ldots,M_{J,n-1,n}$ such that 
$M_{J,k,n} \sim \nu_{J,k}$
with $\nu_{J,k} \in P_p(Sym^+(d))$ and $\E_{\nu_{J,k}}[M_{J,k,n}] = M_{J,k}$ for a certain $p\ge 2$ given in \eqref{E:MomentConditionLogM}. Moreover, to be consistent, 
we assume that the distributions of $M_{J,k,n}$ and $M_{J',k',n'}$ are equal whenever $k2^{-J}$ and $k' 2^{-J'}$ represent the same dyadic number. 
As an example consider the Riemannian signal plus noise model.
\begin{Ex}[Riemannian signal plus noise model]\label{Example:SignalPlusNoise}
Suppose $c\colon [0,1] \to Sym^+(d)$ is an unknown target signal, $c_k = c(k/n)$, $k=0,\ldots,n$ and 
$\xi_k \sim \tilde\nu$ are i.i.d. mean-zero noise random variables in $Sym(d)$, $\tilde\nu \in P^2(Sym(d))$. We will consider the intrinsic signal 
plus noise model as a running example (see also Equation 2.2.17 in \cite{chau2018}) in the following.
\begin{align*}
    X_k &= \Exp_{c_k}(\Gamma_{\Id}^{c_k}(\xi_k))
    = \exp(\log(c_k) + \xi_k) , \quad k=0,\ldots, n.
\end{align*}
\end{Ex}

\textbf{The estimation procedure.} Given input data on a dyadic scale, we obtain a smoothed curve as follows.
\begin{itemize}
\item [(0)] Let $X=(X_k)_{k=0,\ldots,n-1}$ be a sequence on the dyadic scale $J$ such that $2^J = n$.

\item [(1)] Denote $(M_{j,k,n})_{j,k}$ the midpoint pyramid according to (\ref{midpointPyramid}) based on $X$. The empirical wavelet coefficients are $\widehat{D}_{j,k,n}$ for $j\in \{0,\ldots,J\}$ and $k\in \{0,\ldots,2^{j}-1\}$. We obtain these from the (estimated) predicted midpoints via  \eqref{predictionEven}, \eqref{predictionOdd} and \eqref{waveletCoeff} for some refinement order $N=2L+1, L \geq 0$.

In particular, on scale $J$,
$M_{J,k,n} = X_k$ for $k\in \{0,\ldots,n-1\}$.



\item [(2)] The smoothing occurs at scales $j \in \{J_0,\ldots, J\}$, where $J_0 \in \{ 1,\ldots,J \}$ defines the thresholding scale for the empirical wavelet coefficients. We have
\begin{align*}
    \widehat{D}_{j,k,n}^{(t)} = 
    \begin{cases}
    \widehat{D}_{j,k,n} &,j=0,\ldots,J_0-1\\
    0                   &,j=J_0,\ldots,J
    \end{cases}
    \quad \text{for}~k=0,\ldots, 2^{j-1}-1.
\end{align*}
\item [(3)] Starting with $M_{0,0,n}$ and $\widehat{D}_{1,0,n}^{(t)}$ recursively apply the backward wavelet transform to obtain \emph{linearly thresholded estimated
midpoints $(\widehat{M}_{j,k,n})_{j,k}$}. At scale $J$, we obtain from $(\widehat{M}_{J,k,n})_{k}$ the curve $\hat c_n(t) = \sum_k \widehat{M}_{J,k,n} \1_{I_{J,k}}(t)$.\\[5pt]
\end{itemize}

\textbf{The covariance operator.} In order to study asymptotic normality and confidence regions, we need to model the covariance structure of the sampling scheme. We abbreviate the covariance operator of $\log M_{J,k,n}$ by 
$$
	{\cal C}(k 2^{-J}) \colon Sym(d) \to Sym(d).
$$
We assume that $( {\cal C}(u): u\in [0,1] \text{ is dyadic})$ can be extended to a c{\` a}dl{\` a}g process $( {\cal C}(u): u\in [0,1] ) \subset {\cal L}$ whose projections are positive definite in that $\langle A,  {\cal C}(u) A\rangle_F > 0$ for all $u\in[0,1]$. 
Here ${\cal L} = {\cal L}(Sym(d))$ denotes the (Banach-) space of bounded linear operators $T\colon Sym(d) \to Sym(d)$ equipped with the usual induced operator norm
$\| T\|_{\cal L} = \sup\{ \|TA\|_F\mid A \in Sym(d),~ \|A\|_F \leq 1\}$.

Furthermore, we assume a uniform bounded moments condition
\begin{align}\label{E:MomentConditionLogM}
		c_p =	\sup \E[ \| \log M_{j,k,n} \|_F^{p} ] < \infty,
\end{align}
for a $p>2$ and where the supremum is taken over $k \in \{0,\ldots,2^j-1\}$, $j \in \N$. This makes the process $({\cal C}(u) : u\in [0,1] )$ Bochner integrable, viz.,
\begin{align}\label{Cond:CovOp}
		\int_0^1 {\cal C}(u) \diff u \in {\cal L}.
\end{align}
Indeed, the Bochner integral exists if and only if $\int_0^1 \| {\cal C}(u) \|_{{\cal L}} \diff u < \infty$ as Lebesgue integral, this is satisfied in our model by the uniform bounded moments condition from \eqref{E:MomentConditionLogM} (see also Lemma~\ref{L:ExBochnerInt}).

Using that the inner product is bilinear and continuous in each argument, we can exchange the Bochner integral with the Lebesgue integral on any interval $I\subseteq [0,1]$
\[
		\int_I \ls A, {\cal C}(u) B \rs \diff{u} =  \left\ls A, \int_I{\cal C}(u)\diff{u} \  B \right\rs 
\]
for $A,B\in Sym(d)$.
This implies for each $A,B\in Sym(d)$ and $j\le J$
\begin{align}\begin{split}\label{E:ConvergenceCovOp}
		&2^{-(J-j)} \sum_{u=0}^{2^{(J-j)}-1} \ls A,{\cal C} ( k 2^{-j} + u 2^{-J}) B \rs 
		 \\
		&\to  \ 2^j \int_{k 2^{-j}}^{(k+1)2^{-j}} \ls A, {\cal C}(u) B \rs \diff u =\left\ls A, 2^j \int_{k 2^{-j}}^{(k+1)2^{-j}}{\cal C}(u) \diff uB \right\rs 
		\end{split}
\end{align}
    as $J\to\infty$ for $j$ constant. In particular, for a small interval size, we can approximate the integral in \eqref{E:ConvergenceCovOp} by $\ls A, {\cal C}(k 2^{-j}) B \rs$ by the c{\` a}dl{\` a}g property.
    
\begin{Ex}[Continuation of example \ref{Example:SignalPlusNoise}]\label{Example:SignalPlusNoise2} Note that in the signal plus noise model we have $\log M_{J,k,n} = \log (c_k) + \xi_k$ and therefore
\[
    \Cov(\log M_{J,k,n}) = \Cov(\xi_k) = \mathcal{C}(k2^{-J}) \equiv \cal C
\]
for some $\cal C \in {\cal L}$. Hence the extension of the covariance operators is trivially given by $\mathcal{C}(u) := \mathcal{C},~u\in[0,1]$. Suppose the noise matrices 
$\xi_k$ have independent entries, that is 
\[
\xi_k = \sum_{1\leq i\leq j\leq d} z_{ij}^k E_{ij},
\]
where $(z_{ij}^k)_{i,j,k}$ are independent with $\E z_{ij}^k = 0$ and $\Var (z_{ij}^k) = \sigma_{ij}^2$. The matrices $E_{ij} \in Sym(d)$ are defined by 
\[
    (E_{ij})_{k,l} 
    = \begin{cases}
    1 &,\quad  i=k, j=l\ ~\text{or}\ ~ i=l, j=k\\
    0 &, \quad  \text{else}
    \end{cases}     
\]
and obviously form a basis for $Sym(d)$. 

Now $\mathcal{C} = \Cov(\xi_k)\colon Sym(d) \to Sym(d)$ is a $4$th order tensor which can be represented (with respect to 
the standard basis) by an array $(C_{ijnm})_{i,j,n,m}$, i.e. $(\mathcal{C} A)_{i,j} = \sum_{n,m} C_{ijnm}A_{nm}$ for $A \in Sym(d)$. Since the covariance operator
is uniquely characterized by $\ls  \mathcal{C} A,A \rs_F = \Var \ls \xi_k,A \rs_F$ for all $A \in Sym(d)$, writing out this equation yields
\[
    \sum_{i,j} \sum_{n,m} C_{ijnm}A_{nm}A_{ij} 
    = \Var \sum_{i,j} (\xi_k)_{ij} A_{ij} 
    = \sum_{i,j} \Var(z_{ij}^k) A_{ij}^2
    = \sum_{i,j} \sigma_{ij}^2  A_{ij}^2,
\]
which in turn implies for the coefficients 
\[
    C_{ijnm} 
    = \begin{cases}
    \sigma_{ij}^2 &, \quad i=n, j=m\  ~\text{or}\ ~ i=m, j=n\\
    0             &, \quad \text{else}.
    \end{cases} 
\]
A similar consideration shows that in the general case of dependent entries the coefficients of $\cal C$ are given by
\[
    C_{ijnm} = \Cov((\xi_k)_{ij}, (\xi_k)_{nm}) =: \sigma_{ijnm}.
\]
Moreover due to the symmetry constraint $\sigma_{ijnm} = \sigma_{nmij} = \sigma_{mnij} = \sigma_{ijmn}$. 
\end{Ex}

\subsection{Asymptotics of the AI wavelet estimator}\label{Subsec:Asymp_linear_est}

In order to derive our main result on asymptotic normality for the estimator $\wh{M}_{J,k,n}$, it might be useful to first recall existing results on mean-square consistency of this estimator which result from the findings in \cite{chau2018} and \cite{chau2020intrinsic} for a general  class of Riemannian metrics, including our considered log-Euclidean metric. We state here the corresponding result without a proof.


\begin{Prop}[Rate of convergence from \cite{chau2020intrinsic} and \cite{chau2018} Appendix A]\label{P:ConsistencyLinEstimator}
Consider the wavelet estimator $\wh{M}_{J,k,n}$, based on a refinement scheme of order $N$ and on (linear) thresholding on scale $J_0$. Let further $c(t)$ be a smooth curve of SPD matrices, possessing at least $N$ continuous derivatives. Then 
there is a $C\in\R_+$, which does not depend on $n$ and $k$ such that
$$
		\E[ \| \log \wh{M}_{J,k,n} - \E[ \log \wh{M}_{J,k,n}]  \|_F^2 ] \le C 2^{-(J-J_0)}
$$
and
$$
	\| \E[ \log \wh{M}_{J,k,n}] - \log M_{J,k}  \|_F \le C 2^{- N J_0 }.
$$
In particular, the choice $J_0 = \floor{J/(2N+1)}\ $ matches the variance and the squared bias, so
$$
	\E[ d( \wh{M}_{J,k,n}, M_{J,k} )^2 ] = \E[ \| \log \wh{M}_{J,k,n} - \log M_{J,k} \|_F^2 ] \le 2 C n^{-2N/(2N+1)}\ ,
$$
which implies that $\ n^{-1} \sum_{k=0}^{n-1} \E[ d( \wh{M}_{J,k,n}, M_{J,k} )^2 ] = O(n^{-2N/(2N+1)})$.
\end{Prop}

\vspace{0.5cm}

Now we turn to deriving  asymptotic normality of the wavelet estimator which follows from the above lifting scheme and the following proposition.
\begin{Prop}[Asymptotic normality of $\log M_{j,k,n}$]\label{P:NormalityLogM} Assume that the uniform bounded moments condition~\eqref{E:MomentConditionLogM} is satisfied for a $p>2$.
\begin{itemize}\setlength\itemsep{1em}
\item [(a)] Let $j>0$ and let $0\le k \le 2^j-1$ be arbitrary but fixed. Then
$$
	2^{-j/2} n^{1/2} (\log M_{j,k,n} - \E[   \log M_{j,k,n} ] )\Rightarrow {\cal N}(0,\Sigma_{j,k}), \quad n=2^J \to\infty,
$$ 
on the Hilbert space $Sym(d)$ and the covariance operator $\Sigma_{j,k}$ is characterized by the condition
\[
	\Sigma_{j,k} = 2^{j} \int_{k 2^{-j} }^{(k+1)2^{-j}} {\cal C}(u) \diff u .
\]
\item [(b)]  Let $(J_{0,n},k_n)_n$ be a sequence such that $J_{0,n}\to\infty$ and $J-J_{0,n} \to\infty$ and $k_n 2^{-J_{0,n}} \to x \in [0,1]$, which is a point of continuity of ${\cal C}$. Then
$$
	2^{-J_{0,n}/2} n^{1/2} (\log M_{J_{0},k,n} - \E[   \log M_{J_{0},k,n} ] ) \Rightarrow {\cal N}(0,{\cal C}(x)).
$$
\end{itemize}
\end{Prop}
In order to quantify the limiting distribution of the wavelet estimator, we need additionally the limit
\[
		E_{N,\infty} = \lim_{m\to\infty} E_N^m = \Phi_L(0) \in \R^{(2N-1) \times (2N-1)},
\]
see \eqref{E:LimitMatrix}, and the corresponding quantity 
\begin{equation}\label{E:kappa}
 \kappa_N:=   \sum_{i=-2L}^{2L} (E_{N,\infty})^2_{1,i} = \sum_{i=-2L}^{2L} \phi^2_L(i),
\end{equation}
which enters into the variance of the Gaussian approximation. 
\begin{Rem}
A factor of the form $\kappa_N$ is quite common to appear in the limit of wavelet estimators. We point out that for example \eqref{E:kappa} coincides with the term
$w_0^2$ which appears in Theorem 3.4 in \cite{antoniadis1994wavelets}.
\end{Rem}
\begin{Thrm}[Asymptotic normality of the wavelet estimator]\label{T:NormalityLinEst}
Consider the wavelet estimator $\wh M_{J_n,k_n,n}$ at a dyadic position $x \in (0,1)$ of the finest (sample) scale $J_n$, and let $J_{0,n}\to\infty$, $J_n- J_{0,n} \to \infty$, and $k_n 2^{-J_{0,n}} \to x \in [0,1]$. Assume that $\cal C$ is continuous in $x$. Then for $p>2$
\begin{align*}
	 2^{-J_{0,n}/2} n^{1/2} (\log \wh M_{J_n,k_n,n} - \E[\log \wh M_{J_n,k_n,n} ] ) \Rightarrow {\cal N}\Big(0, \kappa_N\ {\cal C}(x)  \Big), \quad n\to\infty.
\end{align*}
\end{Thrm}

As usual with asymptotics ($n\to\infty$) for estimating curves sampled on an asymptotically finer and finer grid, one has to understand the above given expression for the asymptotic variance $\kappa_N C(x)$ in the sense that the reference location $x$ has to remain the same fixed dyadic rational in $(0,1)$ for all $n$.\\
In general, i.e. for $x$ not remaining the same fixed dyadic number, we are still be able to formulate a Central Limit Theorem of the following form:
\begin{Prop}\label{P:Cor_Thm_AN_LE}
Let $x\in (0,1)$ be not necessarily a dyadic number and assume the same regularity conditions as in Theorem~\ref{T:NormalityLinEst}. Then
\begin{align*}
	& Var( \ls A, \log \wh M_{J_n,k_n,n} - \E[\log \wh M_{J_n,k_n,n} ]  \rs ) ^{-1/2} \\
	&\qquad \cdot  (\ls A,\log \wh M_{J_n,k_n,n} - \E[\log \wh M_{J_n,k_n,n} ]\rs ) \Rightarrow {\cal N}(0,1), \quad n\to\infty,
\end{align*}
for any $0\neq A\in Sym(d)$.
\end{Prop}
%

\begin{Rem}
We shed a bit more light onto the fact that for non-dyadic numbers no stable limiting variance of the wavelet estimator may exist, and that we can only give the limiting behavior of the wavelet estimator in the ``standardized'' version of Proposition~\ref{P:Cor_Thm_AN_LE}. This phenomenon has already been observed by  \cite{antoniadis1994wavelets}, Theorem 3.3 and its discussion, for ordinary scalar-valued wavelet estimators. Essentially it is induced by the fact that for non-dyadic numbers on scale $j$ of the form $t_j = 2^j t - [2^j t ] \neq 0$, this sequence $t_j$ - and hence the asymptotic variance expressed as a function of this $t_j$ -  fail to converge. In our situation this can be translated to non-convergence of the matrix products of the type $X_J \cdot X_{J-1} \cdot \ldots \cdot X_{J_0+1}$ with $X_i\in \{ (E_N)_{\dn d}, (O_N)_{\dn d} \}$, a representation occurring in  Lemma~\ref{L:j-stepPrediction} (see also \eqref{E:NormalityLinEst2} in the Supplement). Recall that these products describe what happens if one lifts the AI wavelet estimator smoothed on scale $J_0$ from $J_0$ back to the sampling scale $J$. As both scales $J_0$ and $J$ need to tend to infinity simultaneously (though on different rates with $n$), convergence to finite limit depends on the fact whether there exists a $n_0$ beyond which the sequence  $t_{n,0} := k_n 2^{-J_{0,n}} = x \in (0,1)$ for all $n > n_0$ (case of a dyadic point) or not: The affirmative case amounts to considering $J_0=J_{0,n}$ (i.e. the element $X_{J_0+1}$ in the matrix product) akin to remain fixed for all $n > n_0$, and the considered infinite matrix product possesses a limit (which is given by \eqref{E:LimitMatrix}). In the opposite case $X_{J_0+1}$ continues to change with $n$, leading to no convergence of the matrix product. By the proof of Proposition~\ref{P:Cor_Thm_AN_LE} we can, however, show that the behaviour remains somewhat controlled in this latter case.
\end{Rem}

With our derived Central Limit Theorem~\ref{T:NormalityLinEst} we observe strong parallels in the given rates with those from classical nonparametric curve estimation. Choosing a scale $J_0<J$ is needed in order to collect enough statistical information for the estimation procedure using the midpoint pyramid algorithm, as the variance of a potential estimator on the initial scale $J$ would not decay. In order to increase the precision of the estimator, we also need the usual nonparametric conditions that $J= J_{0,n}\to\infty$ and $J_n- J_{0,n} \to \infty$, for which the precise asymptotic rates can be obtained as usual by matching the (squared) asymptotic bias and variance as in Proposition~\ref{P:ConsistencyLinEstimator}. Note, however, that this optimal choice $J_0 = \floor{J/(2N+1)}$ is merely of asymptotic flavour, as in practice, for reasonably small sample sizes $n=2^J$, this would lead to far too small scales $J_0$. 

\begin{Rem}
Clearly, Theorem~\ref{T:NormalityLinEst} is closely linked to Proposition~\ref{P:NormalityLogM} because the convergence to a Gaussian distribution is induced by the midpoint pyramid algorithm, which is applied in both statements.
The main difference between both results consists in the final lifting step from the coarse scale $J_0$ back to the initial (sampling) scale $J$. This step rescales the limiting covariance operator from ${\cal C}(x)$ to $\kappa_N {\cal C}(x)$. Although we have not been able to prove, we conjecture that $\kappa_N < 1$ if $L=(N-1)/2 > 1$ which shows a decrease of the value of the asymptotic variance. In the following we give a small  discussion of the values of $\kappa_N$ for the first $N=1,3,5,7$:

If $L=0$, then the wavelet estimator simply corresponds to the $\log M_{J_0,k,n}$.
 If $L=1$, i.e. $N=3$, we obtain for the limit of $\lim_{m\to\infty} E_3^m = E_{3,\infty}$ the matrix
\begin{align*}
		E_3^m &=
		\begin{pmatrix}
		\frac{1}{8} & 1 & - \frac{1}{8} & 0 & 0\\
		- \frac{1}{8} & 1 &  \frac{1}{8} & 0 & 0 \\
		0 & \frac{1}{8} & 1 & - \frac{1}{8} & 0 \\
		0 & - \frac{1}{8} & 1 & \frac{1}{8} & 0 \\
		0 & 0 & \frac{1}{8} & 1 & -\frac{1}{8}
	\end{pmatrix}^m \rightarrow 
	\frac{1}{12}
	\begin{pmatrix}
		- 1 & 7 & 7 & - 1 & 0 \\
		- 1 & 7 & 7 & - 1 & 0 \\
		- 1 & 7 & 7 & - 1 & 0 \\
		- 1 & 7 & 7 & - 1 & 0 \\
		- 1 & 7 & 7 & - 1 & 0 
 	\end{pmatrix} = E_{3,\infty}.
\end{align*}

With some effort we were able to analytically compute $\kappa_N$ for some typical choices of $N$.
The results are summarized in Table \ref{table.valuesKappa}. 

\begin{table}
\centering
\begin{tabular}{|c|c|c|c|c|}
\hline
& $N=1$ & $N=3$ & $N=5$ & $N=7$\\
\hline 
&&&&\\ [-1pt]
 $\kappa_N$ & $1$ & $\frac{25}{36} \approx 0.69 $ & $\frac{168549}{213160} \approx 0.79 $ & $\frac{107721892723}{126282847320} \approx 0.85 $\\[5pt]
\hline
\end{tabular}

\caption{}
\label{table.valuesKappa}
\end{table}

\end{Rem}


\section{Confidence sets for the wavelet estimator}\label{Section_ConfidenceSets}

In this section we develop confidence sets for our AI-based wavelet estimators of SPD matrix-valued curves, both based on the asymptotics of the previous section as 
well as on an appropriate form of a wild bootstrap scheme. It is in particular with the latter one that we propose some interesting inference method for non-scalar 
valued non-parametric curve estimators.

\subsection{Asymptotic confidence sets} Our derived asymptotic $(1-\alpha)$-confidence sets are based on Proposition~\ref{P:NormalityLogM} and Theorem~\ref{T:NormalityLinEst}. 
In view of these results we rely on the log-normal distribution which we briefly recall here in the specific context of random matrices in $Sym^+(d)$.

\begin{Def}[Log-normal distribution]\label{D:LogNormalDistribution}
A random matrix $X \in Sym^+(d)$ has a \emph{log-normal} distribution (w.r.t.\ the log-Euclidean metric) with mean $S \in Sym^+(d)$ and covariance
operator $\Sigma \in {\cal L}$ if and only if
\begin{align*}
    \ls \log X,A \rs_F \sim \NN(\ls \log S,A \rs_F, \ls \Sigma A,A \rs_F) \quad \forall A\in Sym(d).
\end{align*}
\end{Def}
Now a straightforward way to obtain distributional properties for log-normal random matrices is to utilize that $Sym(d)$ and $\R^q$, $q=d(d+1)/2$, are isometrically
isomorphic. Consider the mapping $\eta\colon Sym(d) \to \R^q$ defined by
\begin{align*}
    \eta(A) := (A_{11},\ldots,A_{dd}, \sqrt{2}A_{12},\ldots,\sqrt{2}A_{1d},\sqrt{2}A_{23},\ldots,\sqrt{2}A_{2d},\ldots,\sqrt{2}A_{d-1,d})^T
\end{align*}
for $\quad A=(A_{ij})_{ij}$. In other words $\eta$ stacks first the diagonal elements and than (row-wise) the above diagonal elements (scaled by $\sqrt{2}$) into 
a vector. Its easy to see that $\eta$ is one-to-one and $\ls \eta A, \eta B \rs = \ls A,B  \rs_F$, thus, $\eta$ is an isometry. Moreover, a random matrix $X \in Sym^+(d)$ has a \emph{log-normal} distribution in the sense of the definition if and only if
\begin{align*}
    \eta(\log X) \sim \NN (\eta (\log S), \wt{\Sigma}) \quad with~~ \wt{\Sigma} = \Cov(\eta(\log X)) = \eta\Sigma\eta^{-1}.
\end{align*}
This follows immediately since we have for arbitrary $a \in\R^q$ and $A = \eta^{-1}a \in Sym(d)$
\begin{align*}
    \ls \eta(\log X),a \rs = \ls \log X, A \rs_F 
    &\sim \NN( \ls \log S, A \rs_F, \ls \Sigma A,A \rs_F )\\
    &\sim \NN(\ls \eta(\log S),a \rs,\ls \eta \Sigma \eta^{-1}a, a \rs),
\end{align*}
i.e., $\eta(\log X) \sim \NN(\eta(\log S), \tilde \Sigma)$. Applying this insight to Proposition \ref{P:NormalityLogM} (a) yields 
\begin{align*}
2^{-j/2}\sqrt{n}\ \eta(\log M_{j,k,n} - \log M_{j,k})   \Rightarrow \NN(0, \wt \Sigma_{j,k}),
\end{align*}
provided $\E [\log M_{j,k,n}] = \log M_{j,k}$ (which is true under for the Riemannian signal plus noise model from Example~\ref{Example:SignalPlusNoise} and where $\wt\Sigma_{j,k}= \eta\Sigma_{j,k}\eta^{-1}$).

Now suppose that $\wh{\wt\Sigma}_{j,k,n}$ is a consistent estimator for the covariance operator $\wt\Sigma_{j,k}$. As $\wt \Sigma_{j,k}^{-1/2} \NN(0,\wt \Sigma_{j,k})$ agrees with $\NN(0,I_{q\times q})$, we conclude with Slutsky's theorem and the continuous mapping theorem
\begin{align*}
    &2^{-j}n\ \| \wh{\wt\Sigma}_{j,k,n}^{-1/2}\eta(\log M_{j,k,n} - \log M_{j,k}) \|^2\\
    =&  2^{-j}n\ \eta(\log M_{j,k,n} - \log M_{j,k})^T\ \wh{\wt\Sigma}_{j,k,n}^{-1}\ \eta(\log M_{j,k,n} - \log M_{j,k})
    \Rightarrow \chi_q^2.
\end{align*}
Therefore, if we let $\chi_{q,1-\alpha}^2$ denote the $(1-\alpha)$-quantile of the $\chi_q^2$ distribution, it follows that
\begin{align*}
    &CS_{n,1-\alpha}(M_{j,k})\\
    &:= \left\{ S\in Sym^+(d) \mid 2^{-j}n\ \| \wh{\wt\Sigma}_{j,k,n}^{-1/2}\ \eta(\log M_{j,k,n} - \log S) \|^2 
    \leq \chi_{q,1-\alpha}^2 \right\}
\end{align*}
is an asymptotic $(1-\alpha)$-confidence set for $M_{j,k}$.

It is evident that this approach carries over to the wavelet estimator and that an asymptotic confidence set for $ M_{J,k}$ at a dyadic point $x=k2^{-J}=k_n 2^{-J_n}$ is in this case given by
\begin{align}\begin{split}\label{E:ConfidenceSetLinearEstimator}
    &CS_{n,1-\alpha}(M_{J,k})\\
    &:= \left\{ S\in Sym^+(d) ~\Bigg|~ 2^{J_n-J_{0,n}} \kappa_N^{-1} 
    \| \wh{\wt{\mathcal{C}}}(x)^{-1/2} \ \eta(\log \wh{M}_{J_n,k_n,n} - \log S) \|^2 
    \leq \chi_{q,1-\alpha}^2 \right\}
    \end{split}
\end{align}
provided $\wh{\wt{\mathcal{C}}}(x)$ is a consistent estimator for the covariance operator $\wt{\mathcal{C}}(x) = \eta \mathcal{C}(x)\eta^{-1}$. 
However, such an estimator can converge slowly in practice. Moreover, if the difference $J-J_0$ is not substantial the
asymptotic confidence sets are expected to be rather conservative, which is also confirmed by simulations, see Section~\ref{Sec_Applications}.
These reasons motivate to consider confidence sets based on bootstrapping the wavelet estimator as an alternative.

\begin{Ex}[Continuation of example \ref{Example:SignalPlusNoise2}]\label{Example:SignalPlusNoise3}
Recall that in the signal plus noise model $M_{J,k,n} = \exp\{\log (c_k) + \xi_k\}$ we have $\mathcal{C}(u) = \mathcal{C} = \Cov(\xi_0),~u\in[0,1]$. Furthermore, in case
$\xi_k = \sum_{1\leq i\leq j\leq d} z_{ij}^k E_{ij}$, where $(z_{ij}^k)_{i,j,k}$ are independent with $\E[ z_{ij}^k ]= 0$ and $\Var( z_{ij}^k )= \sigma_{ij}^2$,
the coefficients of $\cal C$ (with respect to the standard basis) are simply
\[
    C_{ijnm} 
    = \begin{cases}
    \sigma_{ij}^2, &  i=n, j=m  ~\text{or}~ i=m, j=n\\
    0,             & \text{otherwise}.
    \end{cases} 
\]
Thus, a straightforward calculation reveals 
\begin{align*}
    \wt{\mathcal{C}} &= \eta \mathcal{C} \eta^{-1}
    = \mathrm{diag}(\sigma_{11}^2,\ldots,\sigma_{dd}^2,2\sigma_{12}^2,\ldots,2\sigma_{1d}^2, 2\sigma_{23}^2,\ldots,2\sigma_{2d}^2,\ldots, 2\sigma_{d-1,d}^2 ).
\end{align*}
Now, let $m_{ij}^{(n)} = (\log M_{J_n,k_n,n})_{ij}$ and $a_{ij} = (\log S)_{ij}$. Then writing out the inequality in (\ref{E:ConfidenceSetLinearEstimator}) yields 
\begin{align*}
     2^{J_n - J_{0,n}}\kappa_N^{-1}  
     \sum_{1\leq i \leq j \leq d} \frac{(m_{ij}^{(n)} - a_{ij})^2}{\sigma_{ij}^2} \leq \chi_{q,1-\alpha}^2
\end{align*}
Viewed as an inequality in $(a_{ij})_{i,j}$ this subset of $Sym(d)$ corresponds (via the isometry $\eta$) to an ellipsoid in $\R^q$ with center 
\begin{align*}
(m_{11}^{(n)},\ldots,m_{dd}^{(n)}, \sqrt{2} m_{12}^{(n)},\ldots,\sqrt{2} m_{d-1,d}^{(n)})
\end{align*}
and half-axis 
\begin{align*}
    \tau_{ij} &= 2^{(J_{0,n} - J_n)/2} \sqrt{\kappa_N\sigma_{ij}^2 \chi_{q,1-\alpha}^2},\quad 1\leq i\leq j\leq d.
\end{align*}
If we denote this ellipsoid by $E((m_{ij}^{(n)},\tau_{ij})_{i,j}, N, q, \alpha)$, we arrive at the compact representation 
\begin{align*}
    CS_{n,1-\alpha}(M_{J,k,n}) = \left\{ \exp(A)\mid A = \eta^{-1}(x), x\in E((m_{ij}^{(n)},\tau_{ij})_{i,j}, N, q, \alpha)  \right\}.
\end{align*}
To illustrate, consider the case $d=2$. Then 
\begin{align*}
    CS_{n,1-\alpha}(M_{J,k})
    =& \left\{ \exp\left\{ \begin{bmatrix} a_{11}& a_{12}/\sqrt{2}\\ a_{12}/\sqrt{2}& a_{22} \end{bmatrix} \right\} \Bigg| 
    \begin{pmatrix}
        a_{11}\\a_{22}\\a_{12}
    \end{pmatrix}
    \in
    E((m_{ij}^{(n)},\tau_{ij})_{i,j}, N, 3, \alpha),
    \right\},
\end{align*}
where
\begin{align*}
E((m_{ij}^{(n)},\tau_{ij})_{i,j}, N, 3, \alpha)
=&
\left\{
    \begin{pmatrix}
     m_{11}^{(n)}+r_{11}\sin\theta \cos\varphi\\
    m_{22}^{(n)}+r_{22}\sin\theta \cos\varphi\\
    m_{12}^{(n)}+r_{12}\cos\theta
    \end{pmatrix}
    \Bigg| r_{ij}\in[0,\tau_{ij}], \theta\in [0,\pi], \varphi\in (-\pi,\pi]
\right\}.
\end{align*}
\end{Ex}

\subsection{Wild bootstrap confidence sets} 
The resampling procedure works as follows. Let $J_0^*$ be a base level which is asymptotically coarser than $J$ in the sense that $J - J^*_0 \to \infty$ and 
$J^*_0\to\infty$ as $n\to \infty$. Then generate a bootstrap estimator as follows.
\begin{itemize}\setlength\itemsep{3pt}
	\item [(1)] Given the initial observations $M_{J,k,n}$ obtain the sequence of wavelet estimators $\wh M_{J^*_0,J,k,n}$ w.r.t.\ the base level $J^*_0$.
	\item [(2)] Calculate residuals $\hat \epsilon^*_{J,k,n} = \log \wh M_{J^*_0,J,k,n} - \log M_{J,k,n}$, $k\in \{0,\ldots, 2^J-1\}$.
	\item [(3)] Simulate bootstrap residuals $\epsilon^{boot}_{J,k,n} = \hat \epsilon^*_{J,k,n} V_{J,k,n}$, where the $V_{J,k,n}$ are iid, real-valued random variables with mean zero and unit variance for $k\in \{0,\ldots, 2^J-1\}$ and finite $(2+\delta)$th moment for some $\delta>0$.
	\item [(4)] Create bootstrap observations $ M^{boot}_{J,k,n} = \exp\Big( \log \wh M_{J^*_0,J,k,n} + \epsilon^{boot}_{J,k,n} \Big)$ for each $k\in \{0,\ldots, 2^J-1\}$.
	\item [(5a)] Compute $M^{boot}_{J_0,k,n}$ for $k\in \{0,\ldots, 2^J-1\}$ with the midpoint pyramid algorithm.	
	\item [(5b)] Compute the wavelet estimator $\wh M^{boot}_{J_0,J,k,n}$ w.r.t.\ $J_0$ for $k\in \{0,\ldots, 2^J-1\}$.
\end{itemize}
When simulating the residuals in (3) apart from a standard normal distribution, one can choose the two point distribution
$$
    V = \delta_{- \frac{\sqrt{5}-1}{2}}  \cdot p + \delta_{\frac{\sqrt{5}+1}{2}}  \cdot ( 1 - p ),
$$
with $p = \frac{\sqrt{5}+1}{2\sqrt{5}}$. Then $V$ is the unique centered two point distribution which does not change the second and third moment, viz.,  $\E[V]=0$ and $\E[V^3]=\E[V^2]=1$.\\[5pt]

Now, we consider the two processes
\begin{align*}
	\fM_{J,k,n} &= 2^{-J_0/2} n^{1/2} (\log \wh M_{J_0,J,k,n} - \E[ \log \wh M_{J_0,J,k,n} ]), \quad n\in\N, \\
	\fM_{J,k,n}^{boot} &= 2^{-J_0/2} n^{1/2} (\log \wh M^{boot}_{J_0,J,k,n} - \E[ \log \wh M^{boot}_{J_0,J,k,n} ] ), \quad n\in\N .
\end{align*}

We have the following fundamental result. 
\begin{Thrm}[Validity of the bootstrap]\label{T:ValidityBootstrap}
	Let the moment condition from \eqref{E:MomentConditionLogM} be satisfied for a $p> 4$ and let $n=2^J$ as well as $J_0, J_0^* \to \infty$ such that for some $\delta>0$
\[
	\min\Big\{ (1-J_0^*/J) \frac{p}{2} - (1-J_0/J), (1-J_0/J) \frac{p}{4} \Big\} \ge (1+\delta) \text{ for all } n.
\] 
Let $x\in (0,1)$ be a dyadic rational and let $k/2^{J} = x$ eventually. Then with probability 1
\begin{align*}
    	{\cal L^*}( \fM_{J,k,n}^{boot} ) &\Rightarrow {\cal N}\Big(0,\kappa_N\ {\cal C}(x) \Big), \quad  n\to\infty,
\end{align*}
conditionally on the data $\{M_{J,0,n},\ldots,M_{J,2^J-1,n} \}_n$.
In particular, 
\[
		\sup_{x\in Sym(d)} | \p( \fM_{J,k,n} \le x ) - \p^*( \fM_{J,k,n}^{boot}  \le x) | \to 0 \quad a.s., \quad n\to\infty
\]
conditionally on the data. Here, for a random matrix $X$ with values in $Sym(d)$, we define the distribution function by 
    $\p(X \le x) = \p( X_{i,j} \le x_{i,j},~ \forall~ 1\le i,j\le d)$ for $x\in Sym(d)$.
\end{Thrm}
Note that we need the moment condition from \eqref{E:MomentConditionLogM} to be satisfied for $p$ greater than 4 because this time we need in each setting a law of large numbers to hold for the empirical variance estimate.

Regarding the base level $J_0$, we observe that given $p> 4 (2N+1)/(2N) $, we can choose the optimal $J_0 = \floor{ J/(2N+1)}$ as well as $J_0^* = \floor{ \alpha J}$ for
\[
	\alpha \in \Big(0, 1 - \frac{2}{p}\big(1+\frac{2N}{2N+1} \big) \Big).
\]
Note that the choice $\alpha = 1/(2N+1)$ is admissible in this case. However, given these choices, we have to correct for the asymptotic bias $2^{(J-J_0)/2} (\E[ \log \wh M_{J_0,J,k,n} ] - \log M_{J,k} ) $, which does not vanish in the limit given these choices.\medskip

\begin{Rem}\label{Rem:OverUndersmoothing}
Concerning the choice of the base level $J_0$ we remark the following. 
When compared to the optimal choice $\floor{ J / (2N+1)}$ we distinguish two scenarios which already arise in bootstrap for traditional kernel regression.
If $J_0 - \floor{ J / (2N+1)} \to -\infty$, we oversmooth the estimator; conversely, if $J_0 - \floor{ J / (2N+1)}\to \infty$, we undersmooth. Oversmoothing usually requires implicit or explicit bias correction and has been suggested in many contributions, see for instance the classical papers \cite{hardle1988bootstrapping}, \cite{hardle1991bootstrap}. 
In order to avoid this need for bias correction undersmoothing is recommended (as addressed already in \cite{hall1992bootstrap}, \cite{neumann1994fully}, \cite{hardle2004bootstrap}). In our present setting, and in particular for the simulated data examples in Section~\ref{Sec_Applications}, undersmoothing is also favorable as the optimal choice $J_0 = \floor{ J/(2N+1)}$ would be far too small even for moderate scales $J$.

\end{Rem} 


Based on the above considerations we propose as bootstrap confidence sets
\begin{align}\begin{split}\label{E:BootstrapConfidenceSets}
    CS_{n,1-\alpha}^{boot}(\wh{M}_{J,k}) &:= \{ S\in Sym^+(2) \mid d(\wh{M}_{J,k}, S) < r_k \}\\
    &= \exp\left( \{ A \in Sym(2) \mid  \| \eta(\log \wh{M}_{J,k}) - \eta(A)\| < r_k\} \right)\\
    &= \exp \left( \eta^{-1}\left( \{x\in\R^3 \mid \| \eta(\log \wh{M}_{J,k}) - x \| < r_k \}\right) \right),
    \end{split}
\end{align}
i.e. the metric ball with centre $\wh{M}_{J,k}$ and radius $r_k$ with respect to the log-Euclidean metric. The radius $r_k$ is taken as 
the $(1-\alpha)$-quantile of the distances 
\begin{align*}
    d(\wh{M}_{J,k},~\wh{M}_{J,k,b}^{boot}) = \| \log\wh{M}_{J,k} - \log\wh{M}_{J,k,b}^{boot}\|, \qquad b=1,\ldots,B.
\end{align*}

\section{Numerical simulation examples}\label{Sec_Applications}

In this simulation study we test our proposed estimation procedure for the signal plus noise model 
$M_{J,k,n} = \exp\{ \log(c_k) + \xi_k \} $, where the noise matrices $\xi_k$ have independent Gaussian entries 
$z_{ij}^k \sim \mathcal{N}(0,\sigma_{ij}^2)$, cf. example \ref{Example:SignalPlusNoise2} and \ref{Example:SignalPlusNoise3}. 
We consider the following smooth $Sym^+(2)$-valued curves 
\begin{align*}
    &c_1(t) = \begin{pmatrix}
        50\sqrt{1-(2t)^2} + 0.1 & 2\sin(17\pi t)\\
        2\sin(17\pi t) & 50t + 1
    \end{pmatrix},\\
    &c_2(t) = \begin{pmatrix}
        55\cos(5\pi(t+0.1)/11) & 50\sqrt{\sin(10\pi(t+0.1)/11)/2}\\
        50\sqrt{\sin(10\pi(t+0.1)/11)/2} & 55\sin(5\pi(t+0.1)/11)
    \end{pmatrix},\\
    &c_3(t) = \begin{pmatrix}
        2(5-10t)^2 & 5-10t\\
        5-10t &  1
    \end{pmatrix}
\end{align*}
for $t\in[0,1]$ (see Figure~\ref{Fig:CurvePlots}).\\

For the different curves we choose the following variance parameters
\begin{center}
    \begin{tabular}{llll}
     $c_1\colon$   & $\sigma_{11}= 0.05$, & $\sigma_{22}= 0.1$, & $\sigma_{12}= 0.01$ \\
     $c_2\colon$   & $\sigma_{11}= 0.1$, & $\sigma_{22}= 0.05$, & $\sigma_{12}= 0.1$ \\
     $c_3\colon$   & $\sigma_{11}= 0.1$, & $\sigma_{22}= 0.1$, & $\sigma_{12}= 0.1$
    \end{tabular}
\end{center}

For the different parameters of our AI-wavelet estimators we choose the order of the scheme to be moderate by taking $N=5$, the sample scale $J=10$ (i.e. a sample size of $n=1024$), and the smoothing scale $J_0<J$ to be the only remaining free parameter to be chosen curve by curve. For all simulations being based on $B=100$ bootstrap repetitions (for which we choose $J^*_0=J_0$), we compute first the wavelet estimator
$(\wh{M}_{J,k})_k$ and then the corresponding asymptotic (see \eqref{E:ConfidenceSetLinearEstimator}) and bootstrap (see \eqref{E:BootstrapConfidenceSets}) confidence sets - as derived in the previous section.
Note that in \eqref{E:ConfidenceSetLinearEstimator} we replace the estimated $\widehat{\widetilde{\cal{C}}}$ by the true $\widetilde{\cal{C}}$ (which is known in our simulation set-up) in order to minimize additional variability due to estimation of this variance parameter.

In order to visualize the results we identify positive definite matrices 
\begin{align*}
M = \begin{pmatrix}
    x & z\\
    z & y
\end{pmatrix}
\end{align*}
with points $(x,y,z)$ in the cone $C=\{ (x,y,z)\in\R^3 \mid xy > z^2 \}$. Note that this correspondence is one-to-one. Both
confidence sets are compared with respect to their empirical coverage
based on $K=100$ samples and, since the confidence sets can be considered as subsets in the cone $C$, we
also approximate their respective (Lebesgue-) volumes. 
To compensate for the dependency of the volumes on the location within the cone $C$ - close to the boundary for example the confidence sets are flattened out, and they
generally grow larger with increasing distance from the origin - we also scale them by the volume that corresponds to the matrix
exponential of the unit ball at the respective points.\\
Our simulations show that the wavelet estimator suffers from the classical bias phenomenon of non-parametric curve estimators  at the boundaries. We therefore omit the first and last $100$ values in our computations. Moreover, in case of curve $c_1$ and $c_2$ we observe that our choice of $J_0=7$ and $J_0=5$, respectively, leads to satisfactory empirical coverages but also to a slightly too variable estimator. 
For this we refer to the discussion on undersmoothing in Remark~\ref{Rem:OverUndersmoothing}.
Finally, for curve $c_3$ we also observe a bias to appear around $t=1/2$ (where the curve almost touches the origin) - to be observed in Table~\ref{table:c3} and Figure~\ref{fig_C3}. To demonstrate that resulting empirical coverages for this curve are slightly lower {\em only due to this phenomen}, we also show another simulation set-up in Table~\ref{table:c3-mean-centred} and Figure~\ref{fig_c3_centred}: here we center our confidence sets around the expected value of the wavelet estimator (instead of the true curve), which is approximated - at each considered point in time - by the sample mean over $50$ repetitions.\\[5pt]


\begin{figure}[H]
\centering
\includegraphics[width=120mm, height=60mm]{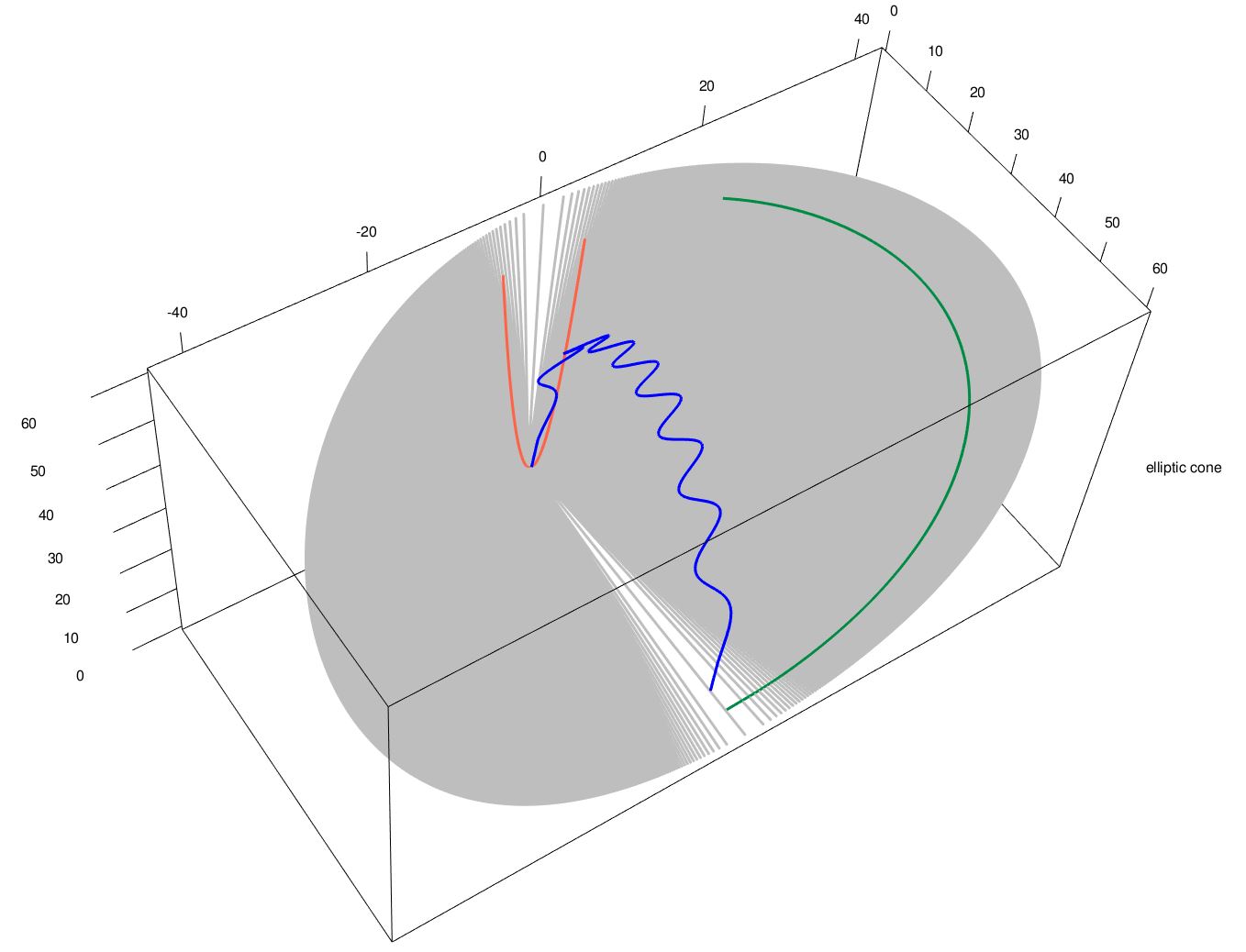}
\caption{Cone $C$ and curves $c_1=$ blue, $c_2=$ green, $c_3=$ red. 
\label{Fig:CurvePlots}}
\end{figure}


\begin{figure}[H]
\centering
\includegraphics[width=100mm, height=45mm]{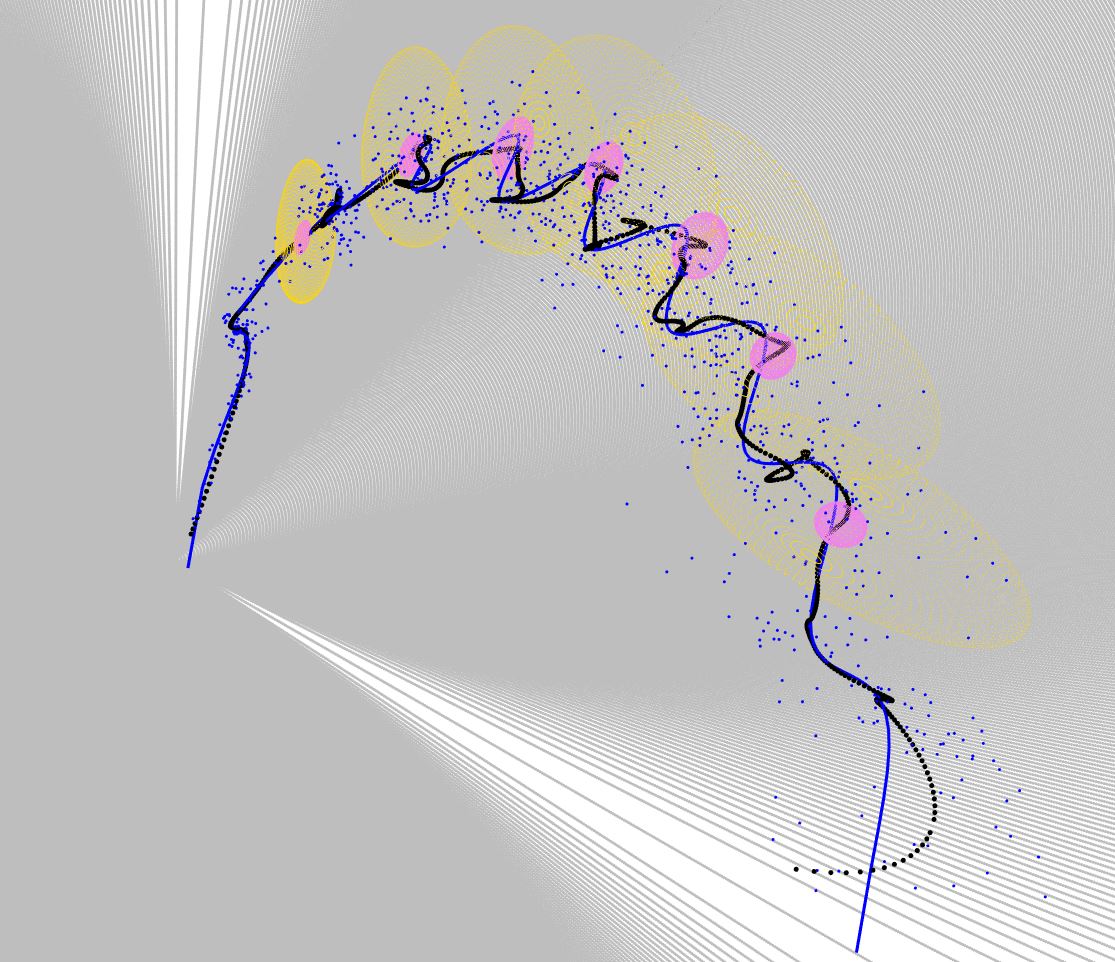}
\caption{AI-wavelet estimator (black), asymptotic (gold) and bootstrap (pink) $0.9$-confidence sets for curve $c_1$ and $J=10$, $J_0=7$. \label{overflow}}
\end{figure}


\begin{figure}[H]
\centering
\includegraphics[width=100mm, height=45mm]{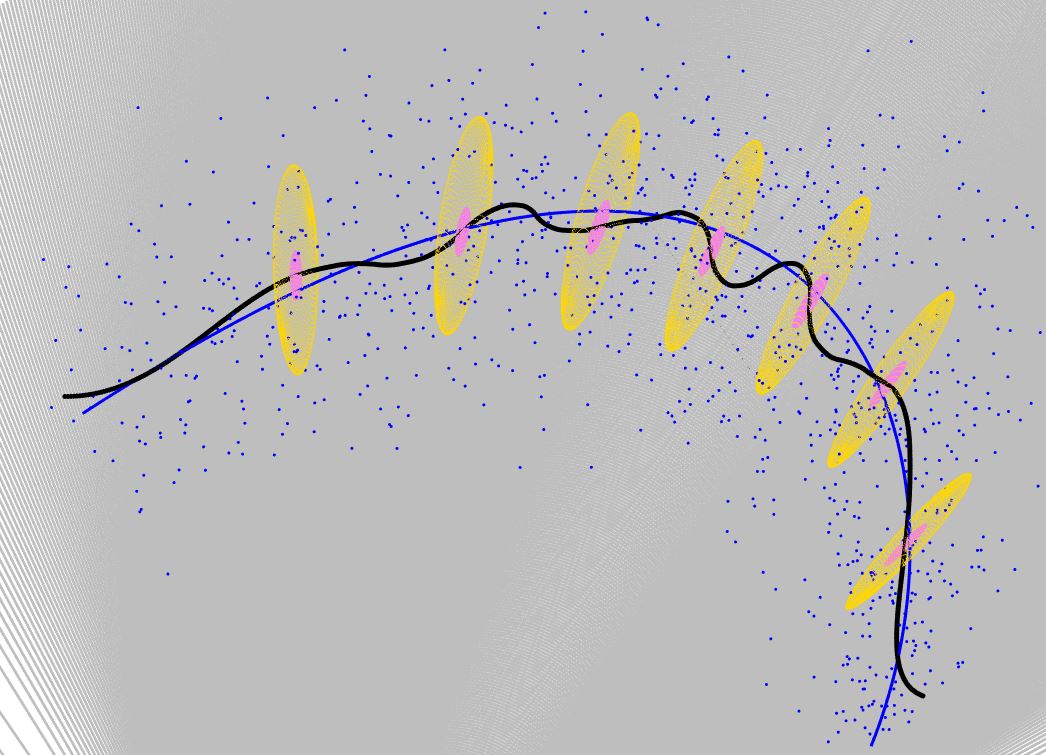}
\caption{AI-wavelet estimator (black), asymptotic (gold) and bootstrap (pink) $0.9$-confidence sets for curve $c_2$ and $J=10$, $J_0=5$. \label{overflow}}
\end{figure}

\begin{figure}[H]
\centering
\includegraphics[width=100mm, height=45mm]{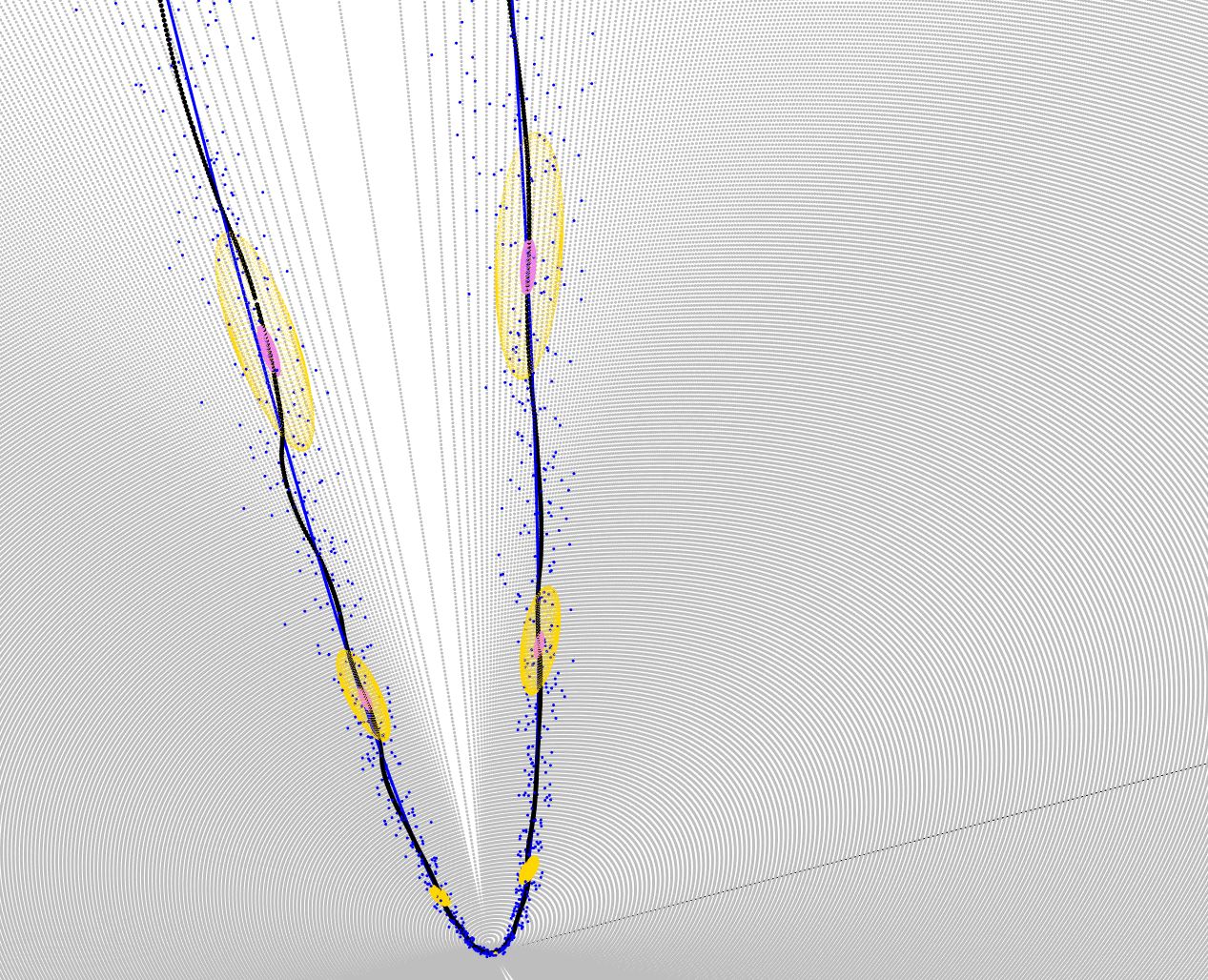}
\caption{AI-wavelet estimator (black), asymptotic (gold) and bootstrap (pink) $0.9$-confidence sets for curve $c_3$ and $J=10$, $J_0=6$. \label{Fig: plot_c3}}
\end{figure}

\begin{table}
\caption{Simulation results for curve $c_1$ and $J=10$, $J_0=7$. The volumes are expressed in the order of magnitude $10^{-4}$. }
\centering
\begin{tabular}[t]{lrrrrrr}
\toprule
\multicolumn{1}{c}{ } & \multicolumn{2}{c}{$\alpha=0.9$} & \multicolumn{2}{c}{$\alpha=0.95$} & \multicolumn{2}{c}{$\alpha=0.975$} \\
\cmidrule(l{3pt}r{3pt}){2-3} \cmidrule(l{3pt}r{3pt}){4-5} \cmidrule(l{3pt}r{3pt}){6-7}
  & Asym. & Boot. & Asym. & Boot. & Asym. & Boot.\\
\midrule
Emp. Cov. & 1.0000 & 0.8809 & 1.0000 & 0.9303 & 1.0000 & 0.9542\\
Scaled Ave. Vol. & 22.3151 & 0.6592 & 31.3509 & 1.0108 & 41.2271 & 1.4380\\
\bottomrule
\end{tabular}
\end{table}

\begin{table}
\caption{Simulation results for curve $c_2$ and $J=10$, $J_0=5$. The volumes are expressed in the order of magnitude $10^{-4}$. }
\centering
\begin{tabular}[t]{lrrrrrr}
\toprule
\multicolumn{1}{c}{ } & \multicolumn{2}{c}{$\alpha=0.9$} & \multicolumn{2}{c}{$\alpha=0.95$} & \multicolumn{2}{c}{$\alpha=0.975$} \\
\cmidrule(l{3pt}r{3pt}){2-3} \cmidrule(l{3pt}r{3pt}){4-5} \cmidrule(l{3pt}r{3pt}){6-7}
  & Asym. & Boot. & Asym. & Boot. & Asym. & Boot.\\
\midrule
Emp. Cov. & 1.0000 & 0.8781 & 1.000 & 0.9298 & 1.0000 & 0.9625\\
Scaled Ave. Vol. & 7.1849 & 0.2328 & 10.0577 & 0.3479 & 13.1792 & 0.5007\\
\bottomrule
\end{tabular}
\end{table}

\begin{table}
\caption{Simulation results for curve $c_3$ and $J=10$, $J_0=6$. The volumes are expressed in the order of magnitude $10^{-4}$.}
\centering
\begin{tabular}[t]{lrrrrrr}
\toprule
\multicolumn{1}{c}{ } & \multicolumn{2}{c}{$\alpha=0.9$} & \multicolumn{2}{c}{$\alpha=0.95$} & \multicolumn{2}{c}{$\alpha=0.975$} \\
\cmidrule(l{3pt}r{3pt}){2-3} \cmidrule(l{3pt}r{3pt}){4-5} \cmidrule(l{3pt}r{3pt}){6-7}
  & Asym. & Boot. & Asym. & Boot. & Asym. & Boot.\\
\midrule
Emp. Cov. & 1.0000 & 0.8197 & 1.0000 & 0.8730 & 1.0000 & 0.9055\\
Scaled Ave. Vol. & 33.2901 & 0.9431 & 46.7046 & 1.3011 & 61.3335 & 1.6766\\
\bottomrule
\end{tabular}\label{table:c3}
\end{table}

\begin{table}

\caption{Simulation results for curve $c_3$ and $J=10$, $J_0=6$. CS centred around the mean estimator. The volumes are expressed in the order of magnitude $10^{-4}$.}
\centering
\begin{tabular}[t]{lrrrrrr}
\toprule
\multicolumn{1}{c}{ } & \multicolumn{2}{c}{$\alpha=0.9$} & \multicolumn{2}{c}{$\alpha=0.95$} & \multicolumn{2}{c}{$\alpha=0.975$} \\
\cmidrule(l{3pt}r{3pt}){2-3} \cmidrule(l{3pt}r{3pt}){4-5} \cmidrule(l{3pt}r{3pt}){6-7}
  & Asym. & Boot. & Asym. & Boot. & Asym. & Boot.\\
\midrule
Emp. Cov. & 1.0000 & 0.8879 & 1.0000 & 0.9434 & 1.0000 & 0.9681\\
Scaled Ave. Vol. & 33.2936 & 1.035167 & 46.7095 & 1.5139 & 61.3399 & 2.0516\\
\bottomrule
\end{tabular}\label{table:c3-mean-centred}
\end{table}


 \begin{figure}[t]
 \centering
\includegraphics[width=120mm, height=30mm]{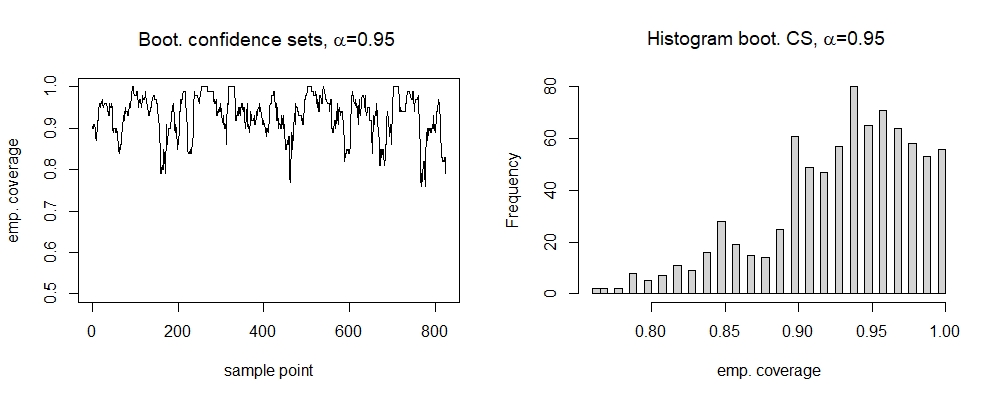}
 \caption{Empirical coverage and corresponding histogram for bootstrap confidence sets for curve $c_1$ and 
 $J=10$, $J_0=7$ \label{overflow}}
\end{figure}

 \begin{figure}[H]
 \centering
\includegraphics[width=120mm, height=30mm]{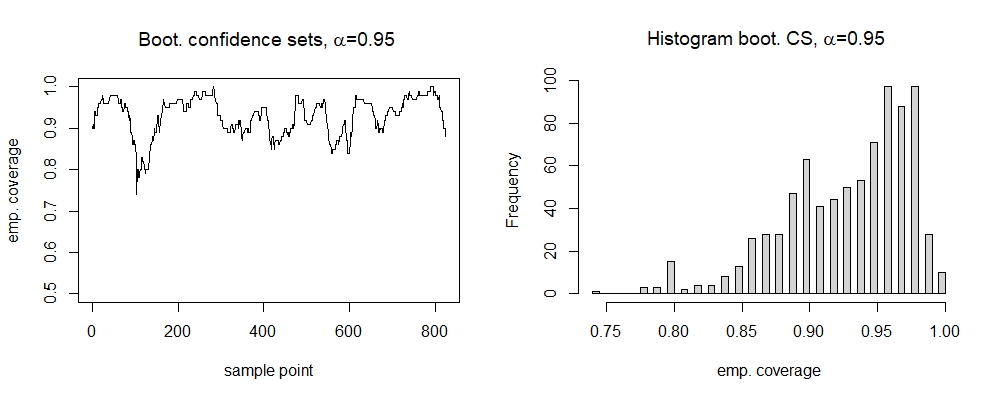}
 \caption{Empirical coverage and corresponding histogram for bootstrap confidence sets for curve $c_2$ and 
 $J=10$, $J_0=5$ \label{overflow}}
\end{figure}
    
 \begin{figure}[H]
 \centering
\includegraphics[width=120mm, height=30mm]{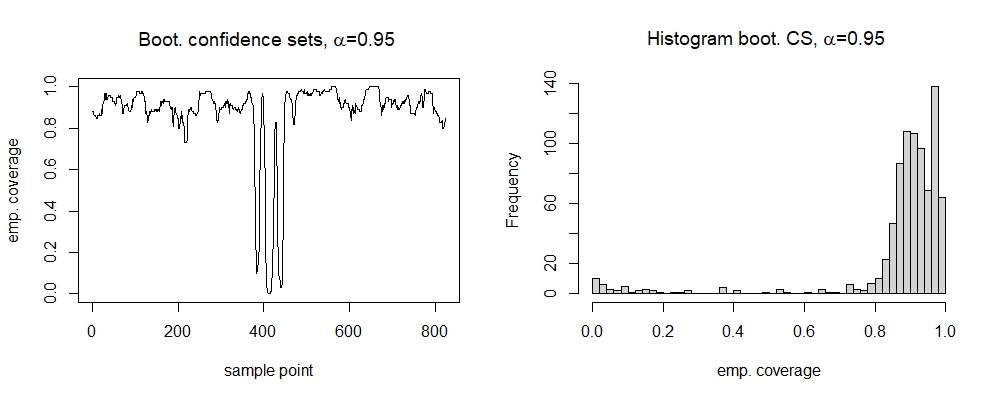}
 \caption{Empirical coverage and corresponding histogram for bootstrap confidence sets for curve $c_3$ and 
 $J=10$, $J_0=6$ \label{fig_C3}}
\end{figure}

 \begin{figure}[H]
 \centering
\includegraphics[width=120mm, height=30mm]{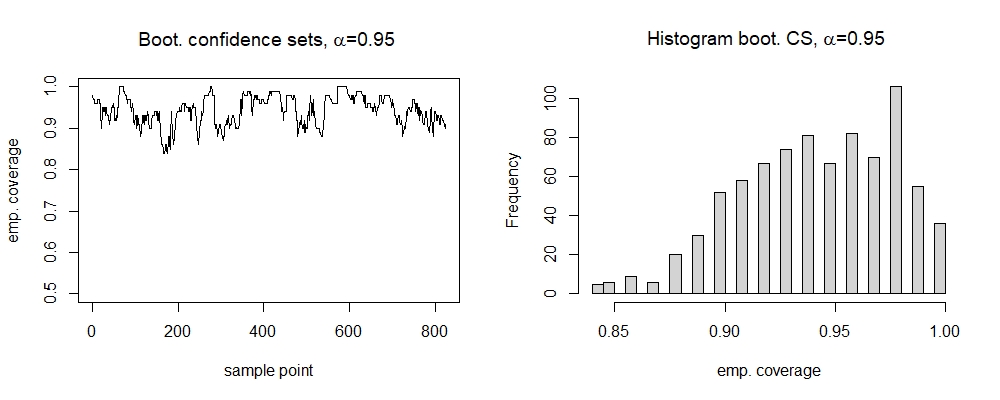}
 \caption{Empirical coverage and corresponding histogram for bootstrap confidence sets centred around 
 the mean estimator for curve $c_3$ and $J=10$, $J_0=6$ \label{fig_c3_centred}}
\end{figure}

Overall, we observe that for our bootstrap-based confidence sets not only our empirical coverages get close to the nominal ones on average (over all considered time points outside the boundaries) but also do we observe only a small proportion where this is not the case. This is likely to be the case due to the usual potential bias problem (as discussed above) which is inherent to any nonparametric curve estimation procedure - and not specific to our more challenging matrix-valued estimation problem. Quite satisfactorily, we also observe that the volumes of our bootstrap-based confidence sets are (on average) up to two orders of magnitude smaller than those that are given by the asymptotic approach. Recall that asymptotic confidence sets are known to be most often too conservative, which is also the reason for their empirical coverages observed here to be very close to $1$ (as already mentioned previously below equation~\eqref{E:ConfidenceSetLinearEstimator}).

\section{Conclusions}\label{Sec:Conclusions}

In this paper we have been proposing two constructions of (pointwise) confidence sets for the (linear) wavelet estimator of SPD-matrix valued curves. Whereas mean-square consistency of our estimator on a Riemannian manifold, constructed via Average Interpolation wavelet schemes, is known to hold since the work by \cite{chau2020intrinsic}, no such result on statistical inference has been existing before. In the first place, we have proven a Central Limit Theorem to hold for our linear wavelet estimator, which enables us to propose asymptotic confidence regions. These turned out to be rather conservative, as it is usually the case in nonparametric curve estimation problems, and in particular for curves living on curved, non-Euclidean manifolds. To circumvent this, we developed (wild) bootstrap confidence schemes which we showed to be theoretically valid, and empirically very satisfying: simulation studies confirmed these second sets to be far less conservative while keeping the nominal level. 

Quite naturally, the question arises about inference for the non-linear version of our wavelet estimator, as proposed again by \cite{chau2020intrinsic}, based on trace-thresholding of the matrix-valued empirical wavelet coefficients. Constructing confidence sets for this type of non-linear estimator is already a non-trivial task in the classical Euclidean set-up of first-generation wavelets for scalar curves, the related literature is sparse (e.g., \cite{chau2016functional}, using \cite{robins2006adaptive}). So although we believe that our approach based on the bootstrap sets can be successfully applied to the threshold wavelet estimator, we decided to leave this and further investigations on this question for future work. Note, however, that as a first theoretical step into this direction, we were able to derive a Central Limit Theorem, analogous to our Theorem~\ref{P:NormalityLogM}, for the (whitened) empirical wavelet coefficients \eqref{whitenedWaveletCoeff}.

For the construction of our wavelet estimators, we used the log-Euclidean metric, which seems nicely taylored to the problem of denoising curves of SPD-matrix values. 
The log-Euclidean metric is among the computationally easiest constructions of Riemannian metrics for the space of positive symmetric matrices, replacing Euclidean distances by ones that are appropriate for this curved manifold. Moreover, our used metric enjoys important properties such as invariance with respect to unitarian congruent transformations. This implies in particular permutation-equivariant estimators of covariance matrices of multivariate or time series data (which is not the case for the often used Cholesky approach). Finally, and importantly for the constructions of correct confidence regions, this metric avoids any undesirable swelling effect: any element of the given confidence region remains a positive-definite symmetric matrix. Future work could be interesting in order to study whether our treated programme for inference could be developed for other matrix-valued objects, such as curves of positive semi-definite matrices. 

\addcontentsline{toc}{section}{Acknowledgments}
\textbf{Acknowledgements.} Johannes Krebs gratefully acknowledges the support of the Deutsche Forschungsgemeinschaft (grants KR-4977/1-1, KR-4977/2-1) and the hospitality of ISBA/LIDAM (UCLouvain).

\bibliography{Bibliography}


\newpage

{\Large\textbf{Supplement}}

\section{Technical details on Section~\ref{Section_AI}}\label{Section_Proofs}

\subsection{Details on the AI refinement scheme of Section~\ref{Subsec:AI-RefinementScheme}}\label{S:DetailsMidpointPrediction}

\textit{Intrinsic polynomial interpolation.}
The fundamental tool for interpolation is Neville's algorithm (\cite{ma2012Neville}, chapter 9.2) which can immediately be applied to our setting in the log-Euclidean metric. We only briefly sketch the nature of this interpolation scheme, more illustrations can be found in, again, \cite{chau2020intrinsic}. For given tupels $(x_0,P_0),\ldots, (x_n,P_n) \in \R \times Sym^+(d)$ with $x_i < x_{i+1}$, let $p_{i,i}(x) = P_i$ for $x\in \R$ and $i=0,\ldots,n$. Now iteratively
define
\begin{align}\begin{split}\label{IntrinsicPolynomialInterpolation}
	p_{i,j}(x)
	&= \Exp_{p_{i,j-1}(x)} \left( \frac{x-x_i}{x_j - x_i} \Log_{p_{i,j-1}(x)}(p_{i+1,j(x)})\right)\\
	&= \exp\left\{\frac{x_j - x}{x_j-x_i} \log(p_{i,j-1}(x)) + \frac{x-x_i}{x_j - x_i} \log(p_{i+1,j}(x)) \right\}, 
\end{split}
\end{align}
$0\leq i<j\leq n$, $x\in\R$.
At the final iteration $p_{0,n}\colon \R \to Sym^+(d)$ is the \textit{intrinsic polynomial} (of degree $n$) interpolating $P_0,\ldots,P_n$ at $x_0,\ldots, x_n$.\\

\textit{Derivation of \eqref{Eq:AI-Refinement1}}:
\begin{align*}
    &\mathrm{Mean}(\rho; (k-L)2^{-j}, (2k+1)2^{-(j+1)})\\
    =& \exp\left\{ \frac{1}{(2L+1)2^{-(j+1)}} \left(\int_{(k-L)2^{-j}}^{k2^{-j}} \log(\rho(t)) \, \dd t
    + \int_{k2^{-j}}^{(2k+1)2^{-(j+1)}} \log(\rho(t)) \, \dd t  \right) \right\}\\
    =& \exp \Bigg\{ \frac{L2^{-j}}{(2L+1)2^{-(j+1)}} \log\circ\exp \left\{ \frac{1}{L2^{-j}}  \int_{(k-L)2^{-j}}^{k2^{-j}} \log(\rho(t)) \, \dd t \right\}\\
    &+ \frac{1}{2L+1} \log\circ\exp \left\{ \frac{1}{2^{-(j+1)}} \int_{k2^{-j}}^{(2k+1)2^{-(j+1)}} \log(\rho(t)) \, \dd t \right\} \Bigg\}\\
    =& \exp \left\{ \frac{2L}{2L+1}\log( \mathrm{Mean}(\rho; (k-L)2^{-j}, k2^{-j} )) + \frac{1}{2L+1}\log(\wt{M}_{j+1,2k}) \right\}\\
    =& \gamma\left(\frac{1}{2L+1}; \ol{M}_{j,L}, \wt{M}_{j+1,2k} \right).
\end{align*}

\textit{Derivation of \eqref{predictionEven}}:
Substituting the intrinsic interpolation polynomial into \eqref{Eq:AI-Refinement1} we have 
\begin{align*}
    \pi( (2k+1)2^{-(j+1)} ) = \exp\left\{ \frac{2L}{2L+1}\log( \ol{M}_{j,L}) + \frac{1}{2L+1}\log(\wt{M}_{j+1,2k})  \right\}
\end{align*}
which can be rearranged as 
\begin{align*}
    &\gamma( -2L; \pi((2k+1)2^{-(j+1)}), \ol{M}_{j,L})\\
    =& \exp\left\{ (2L+1)\log(\pi((2k+1)2^{-(j+1)}) - 2L \log(\ol{M}_{j,L})) \right\}\\
    =& \wt{M}_{j+1,2k}.
\end{align*}

\begin{proof}[Details on Equations~\eqref{E:FastPredictEx1} and \eqref{E:FastPredictEx2}]
First, using \eqref{IntrinsicPolynomialInterpolation}, we can write out the first few iterations explicitly as follows:
\begin{align}
    p_{i-1,i}(x) &= \exp\left\{ \log(P_{i-1}) + \frac{x-x_{i-1}}{x_i - x_{i-1}}\Big(  \log(P_i) - \log(P_{i-1}) \Big) \right\}, \quad  \nonumber
  \end{align}
  for $i\in\{1,\ldots, n\}$ and 
  \begin{align}
    \begin{split}\label{E:IntrinsicPolynomialInterpolationP02}
      &p_{i-2,i}(x) = \exp \Bigg\{ \log(P_{i-2})  + \Big( \frac{x-x_{i-2}}{x_{i-1} - x_{i-2}} - \frac{x-x_{i-2}}{x_i - x_{i-2}} \cdot   \frac{x-x_{i-2}}{x_{i-1}- x_{i-2}}  \\
      &\qquad\qquad\qquad\qquad  + \frac{x-x_{i-2}}{x_i - x_{i-2}} \Big) \Big(  \log(P_{i-1}) - \log(P_{i-2}) \Big) \\
	&\qquad\qquad\qquad\qquad+ \frac{x-x_{i-2}}{x_i - x_{i-2}}\cdot \frac{x-x_{i-1}}{x_i -x_{i-1}} \Big( \log(P_{i}) - \log(P_{i-1})  \Big)  \Bigg\},                 \end{split}  
\end{align}
for $i\in\{2,\ldots,n\}$.

Next, denote $\wh{M}((2k+1)2^j)$ the intrinsic polynomial interpolating the points $(k2^{-(j-1)}, \ol{M}_{j-1,1})$, $((k+1)2^{-(j-1)}, \ol{M}_{j-1,2})$ and
$((k+2)2^{-(j-1)}, \ol{M}_{j-1,3})$. Evoking the explicit formula for $p_{0,2}$ in \eqref{E:IntrinsicPolynomialInterpolationP02}, we get
\begin{align*}
    &\wh{M}((2k+1)2^j)\\
    =& \exp\Bigg\{ \log(\ol{M}_{j-1,1}) 
        + \left( \frac{1}{2} - \frac{1}{4}\cdot\frac{1}{2} + \frac{1}{4} \right)\Big(\log(\ol{M}_{j-1,2}) - \log(\ol{M}_{j-1,1})\Big)\\ 
        &\quad\qquad - \frac{1}{4}\cdot\frac{1}{2}\Big(\log(\ol{M}_{j-1,3}) - \log(\ol{M}_{j-1,2})\Big) \Bigg\}\\
    =& \exp\left\{ \frac{17}{24}\log(M_{j-1,k-1}) + \frac{1}{3}\log(M_{j-1,k}) - \frac{1}{24}\log(M_{j-1,k+1}) \right\}.
\end{align*}
Therefore,
\begin{align*}
    \wt{M}_{j,2k} 
    &= \exp\left\{ 3\log(\wh{M}((2k+1)2^j)) - 2\log(\ol{M}_{j-1,1}) \right\}\\
    &= \exp\left\{ \frac{1}{8}\log(M_{j-1,k-1}) + \log(M_{j-1,k}) - \frac{1}{8}\log(M_{j-1,k+1}) \right\}
\end{align*}
as well as
\begin{align*}
    \wt{M}_{j,2k+1}
    &= \exp\left\{ 2\log(M_{j-1,k}) - \log(\wt{M}_{j,2k}) \right\}\\
    &= \exp\left\{ -\frac{1}{8}\log(M_{j-1,k-1}) + \log(M_{j-1,k}) + \frac{1}{8}\log(M_{j-1,k+1}) \right\}.
\end{align*}
\end{proof}

\subsection{Fast prediction}

The operation $(\cdot)_{\dn d}$ commutes with multiplication.
\begin{Lem}\label{L:CommDiagExpansion}
Let $A\in\R^{m\times n}$ and $B\in\R^{n\times p}$ and $d\in\N_+$. Then $(AB)_{\dn d} = A_{\dn d} B_{\dn d}$. In addition, if $A$ is non singular, then $ (A_{\dn d})^{-1} = (A^{-1})_{\dn d} $.
\end{Lem}
\begin{proof}
Using the multiplication rules for block matrices, we have
\begin{align*}
	A_{\dn d} B_{\dn d} &=  \Big(  \Big(  \sum_{k=1}^m ( a_{i,k} I_d ) (b_{k,j} I_d )  \Big)_{i,j}  \Big)_{i,j} =  \Big(  \Big( \sum_{k=1}^m a_{i,k}  b_{k,j}  I_d   \Big)_{i,j}  \Big) _{i,j} \\
	&= \Big[ \Big( \Big(\sum_{k=1}^m a_{i,k} b_{k,j}  \Big)_{i,j}  \Big)_{i,j}  \Big]_{\dn d}  = (AB)_{\dn d} .
\end{align*}
Moreover, if $A^{-1}$ exists, then
$$
	I_{nd} = (I_n)_{\dn d} = ( A A^{-1})_{\dn d} = A_{\dn d} (A^{-1})_{\dn d},
$$
i.e., $(A^{-1})_{\dn d}$ is the inverse of $A_{\dn d}$.
\end{proof}

This lemma implies the following structural result regarding the characteristic polynomial of a square matrix.
\begin{Cor}
Let $A\in\R^{n\times n}$ with Jordan decomposition $A=S J S^{-1}$ and characteristic polynomial $\chi_A(\lambda)= det( A - \lambda I_n)$. Then 
$$
	A_{\dn d} = S_{\dn d} J_{\dn d} (S_{\dn d})^{-1}
$$
and $\chi_{A_{\dn d}}(\lambda) = (\chi_A (\lambda))^d$.
\end{Cor}
\begin{proof}
First using Lemma~\ref{L:CommDiagExpansion}, we see
$$
	A_{\dn d} = (S J S^{-1})_{\dn d} = S_{\dn d} J_{\dn d} (S^{-1})_{\dn d} = S_{\dn d} J_{\dn d} (S_{\dn d} )^{-1}.
$$
Consequently, using Lemma~\ref{L:CommDiagExpansion} another time, the characteristic polynomial of $A_{\dn d}$ satisfies
\begin{align*}
	\chi_{A_{\dn d}}(\lambda) &= \det( (SJ S^{-1})_{\dn d} - \lambda I_{nd} ) = \det( S_{\dn d} (J_{\dn d} - \lambda I_{nd}) (S_{\dn d})^{-1} ) \\
	&= \det(J_{\dn d} - \lambda I_{nd} ) = \det( J - \lambda I_n )^d = \chi_{A}(\lambda)^d.
\end{align*}
This completes the proof.
\end{proof}

\begin{proof}[Proof of Lemma~\ref{L:j-stepPrediction}]
The assertion obviously holds for $m=1$.\\
$m \to m+1:$ Let $p\in \{ 0,\ldots, 2^{m+1}-1 \}$. Suppose $p$ is even. Then $p=2q$ for a $q\in \{ 0,\ldots, 2^m-1 \}$. By (\ref{EvenTransition}) we get
\begin{align*}
    \begin{bmatrix}
    \log \wt{M}_{j+m+1,2^{m+1}k-2L+p}\\
    \vdots\\
    \log \wt{M}_{j+m+1,2^{m+1}k+2L+p}
    \end{bmatrix}
   & =
    \begin{bmatrix}
    \log \wt{M}_{j+m+1,2(2^{m}k+q)-2L}\\
    \vdots\\
    \log \wt{M}_{j+m+1,2(2^{m}k+q)+2L}
    \end{bmatrix}
    \\
    &= 
    (E_N)_{\dn d}
    \begin{bmatrix}
    \log \wt{M}_{j+m,2^{m}k+q-2L}\\
    \vdots\\
    \log \wt{M}_{j+m,2^{m}k+q+2L}
    \end{bmatrix}
    \\
    &=
   ( E_N )_{\dn d} \cdot X_m \cdot\ldots\cdot X_1 
    \begin{bmatrix}
    \log M_{j,k-2L}\\
    \vdots\\
    \log M_{j,k+2L}
    \end{bmatrix}.
\end{align*}
Suppose $p$ is odd, than $p=2q+1$ for $q\in \{ 0,\ldots, 2^m-1 \}$. Using (\ref{OddTransition}) a similar calculation yields the assertion.
\end{proof}

\section{Technical details on wavelet regression in Section~\ref{Section_WaveletRegression}} \label{Sec: Proofs_Wav_Reg}

\begin{Lem}\label{L:ExBochnerInt}
Let $u = k 2^{-J}$ be a dyadic rational in $[0,1]$, then
\[
		\| {\cal C}(u) \|_{\cal L} \le \E[ \| \log M_{J,k,n} \|_F^2 ].
\]
In particular, $\sup_{u\in [0,1]} \| {\cal C}(u) \|_{\cal L} \le \sup \E[ \| \log M_{J,k,n} \|_F^2 ]$, where the supremum on the right-hand side is taken over all dyadic rationals.
\end{Lem}
\begin{proof}
Let $\epsilon>0$. For each $u\in [0,1]$, there is an $A_u \in Sym(d)$, $\|A_u\|_F =1$, such that 
$$\| {\cal C}(u) \|_{\cal L} - \epsilon \le \|{\cal C}(u) A_u \|_F = \ls {\cal C}(u) A_u, {\cal C}(u) A_u \rs^{1/2}.  $$
Let $u = k 2^{-J}$ be a dyadic rational for $n= 2^J$. The square of the right-hand side of the last inequality equals by definition
\begin{align*}
		&\Cov( \ls {\cal C}(u) A_u, \log M_{J,k,n} \rs, \ls A_u, \log M_{J,k,n} \rs ) \\
		&\le \Var(\ls {\cal C}(u) A_u, \log M_{J,k,n} \rs)^{1/2} \Var( \ls A_u, \log M_{J,k,n} \rs )^{1/2}.
\end{align*}
Obviously, $\Var( \ls A_u, \log M_{J,k,n} \rs ) \le \E[ \| \log M_{J,k,n} \|_F^2 ]$. Also,
\begin{align*}
		&\Var(\ls {\cal C}(u) A_u, \log M_{J,k,n} \rs) \\
		&\le \E[ \ls {\cal C}(u) A_u, \log M_{J,k,n} \rs^2 ] \\
		&= \E[ \Cov( \ls A_u, \log M_{J,k,n} \rs, \ls B, \log M_{J,k,n} \rs)^2 |_{B = \log M_{J,k,n} } ] \\
		&\le \E[ \Var(  \ls A_u, \log M_{J,k,n} \rs ) \ \Var( \ls B, \log M_{J,k,n} \rs ) |_{B = \log M_{J,k,n} }  ] \\
		&\le  \Var(  \ls A_u, \log M_{J,k,n} \rs ) \ \E[ \E[\ls B, \log M_{J,k,n} \rs^2]|_{B = \log M_{J,k,n} }   ] \\
		&\le \E[ \| \log M_{J,k,n} \|_F^2 ] \ \E[ \E[ \|B\|_F^2 \|\log M_{J,k,n}\|_F^2 ] |_{B = \log M_{J,k,n} }   ] \\
		&= \E[ \| \log M_{J,k,n} \|_F^2 ]^3.
\end{align*}
Thus, $\| {\cal C}(u) A_u \|_F \le \E[ \| \log M_{J,k,n} \|_F^2 ]$. 

The amendment follows from the c{\` a}dl{\` a}g property of ${\cal C}$.
\end{proof}

In order to prove the asymptotic normality, we rely on Lyapunov's condition.
\begin{Thrm}[Lyapunov condition]\label{Lyapunov}
Let $(\xi_{n,i})_{i=1}^{\ell_n}$ be a sequence of independent real-valued random variables for each $n\in\N$ such that $\ell_n\le cn$ for some $c>0$. Assume $\lim_{n\to\infty} n^{-1} \sum_{i=1}^{\ell_n} Var(\xi_{n,i}) = \sigma^2 < \infty$ and $\sup_n \sup_i \E[ |\xi_{n,i}|^{2+\delta}]< \infty$ for some $\delta>0$. Then $n^{-1/2} \sum_{i=1}^{\ell_n}(\xi_{n,i}-\E[\xi_{n,i}]) \Rightarrow {\cal N}(0,\sigma^2)$ as $n\to\infty$.
\end{Thrm}

\begin{proof}[Proof of Proposition~\ref{P:NormalityLogM}]
To keep the notation simple, we omit the dependence of the quantities $J,j,J_0,k$ on $n$.
We split the proof in two parts, the first part covering the statement for the fixed dyadic number $k 2^{-j}$, the second the statement for $k 2^{-j}\to x$.

\textit{(a)} We rely on the Cramér-Wold-device. Let $A\in Sym(d)$ be arbitrary but fixed. In the first step, we derive the limiting expression of the covariance operator.  A representation of $\log M_{j,k,n}$ in terms of the $\log M_{J,u,n}$ follows from the midpoint pyramid algorithm because
\begin{align}\label{E:NormalityLogM1}
	\log M_{j,k,n} = 2^{-(J-j)} \sum_{u=0}^{2^{(J-j)}-1} \log M_{J,u+k 2^{(J-j)} , n}.
\end{align}
In particular, using the independence of the $(M_{J,k,n})_k$, we obtain from \eqref{E:NormalityLogM1}
\begin{align}
	2^{-j} n \  \Var( \ls A, \log M_{j,k,n} \rs )
 &= 2^{-(J-j)} \sum_{u=0}^{2^{(J-j)}-1} \Var( \ls A, \log M_{J,u + k 2^{(J-j)},n } \rs )\nonumber \\
	&= 2^{-(J-j)} \sum_{u=0}^{2^{(J-j)}-1} \ls A, {\cal C}( u2^{-J} + k2^{-j}) A \rs , \label{E:NormalityLogM2}
	\end{align} 
where ${\cal C}( u 2^{-J} + k2^{-j}) $ is the covariance operator of $\log M_{J,u + k 2^{(J-j)},n }$. Using the convergence result from \eqref{E:ConvergenceCovOp}, the right-hand side of \eqref{E:NormalityLogM2} converges to 
$$
	c^* = 2^j \int_{k 2^{-j} }^{ (k+1) 2^{-j} } \ls A, {\cal C}(u) A \rs \diff u > 0.
$$	

The Lyapunov condition from Theorem~\ref{Lyapunov} can now be established with the help of this last convergence result and the uniform bounded moments condition of the curve from \eqref{E:MomentConditionLogM}. The latter implies
\begin{align}\begin{split}\label{E:NormalityLogM3}
		 \E[ | \ls A, \log M_{J,u,n} \rs |^{2+\delta} ]	&\le \| A \|_F^{2+\delta} \ \sup \E[ \| \log M_{J,u,n} \|_F^{2+\delta} ]  < \infty, 
\end{split}\end{align}
where the last supremum is taken over $u\in \{0,\ldots,2^J-1 \}$ and $J\in \N_+$. So, the requirements of Theorem~\ref{Lyapunov} for the array $(\ls A, \log M_{J,u + k 2^{J-j},n} \rs)_u$ are satisfied and $2^{-j/2} n^{1/2} (\log M_{j,k,n} - \E[\log M_{j,k,n}])\Rightarrow {\cal N}(0,c^*)$ as $n\to\infty$. This shows (a).

\textit{(b)} Let $k 2^{-J_0}\to x$. We verify the limit of the variance using  \eqref{E:NormalityLogM2}. We have
\begin{align*}
	| \ls A, {\cal C}(k2^{-J_0} + u2^{-J}) A \rs  -  \ls A, {\cal C}(x) A \rs | \le \| A \|_F^2 \  \| {\cal C}(k2^{-J_0} + u2^{-J}) -   {\cal C}(x) \|_{\cal L}.
\end{align*}
Plainly, $| k2^{-J_0} + u2^{-J} - x | \le | k2^{-J_0} - x | + 2^{-J_0} \rightarrow 0$. Using the continuity of $\cal C$ in the point $x$, 
$$
	\lim_{n\to\infty} 2^{-J_0} n \  \Var( \ls A, \log M_{J_0,k,n} \rs ) = \langle A, {\cal C}(x) A \rangle > 0.
$$
The Lyaponov condition for the sequence $(J_{0,n},k_n)_n$ can be verified in a similar spirit as in (a). This completes the proof.
\end{proof}

For proving Theorem~\ref{T:NormalityLinEst}, in order to study the AI wavelet estimator at a specific point $x\in [0,1]$, we introduce some terminology regarding the representation of dyadic numbers in the unit interval. Let $x\in[0,1]$ have the dyadic representation $(d_0,d_1,\ldots)$, that is
$
		x = \sum_{i=0}^\infty 2^{-i} d_i.
$
We agree to choose the shortened form for dyadic rationals, e.g., we choose the representation 
$
(d_0,d_1,\ldots,d_{j},1,0,0,\ldots,)
$
 for the number $x = \sum_{i=0}^j 2^{-i} d_i + 2^{j+1}$. (We could use the representation 
 $
 	(d_0,d_1,\ldots,d_j,0,1,1,\ldots)
$ instead, this does not change the arguments much).

 Define for a scale $j\in\N_+$ the dyadic approximation of $x$ by the first $j$ entries $x^{(j)} = (d_0,d_1,\ldots,d_j)$. This means, the dyadic rational $x^{(j)}$ is a ratio of the form $k^{(j)} 2^{-j}$, where
 \begin{align}
 		k^{(j)} &= \sum_{i=0}^j  2^i d_{j-i} = 2^0 d_j + 2 d_{j-1} + \ldots +2^{j-j'-1}  d_{j'+1}  + 2^{j-j'} \sum_{i=0}^{j'} 2^{i} d_{j'-i} \nonumber \\
		&= 2^0 d_j + 2 d_{j-1} + \ldots +2^{j-j'-1}  d_{j'+1}  + 2^{j-j'} k'_{j,j'}, \label{E:DyadicRepr1}
 \end{align}
for a smaller scale $j'\le j$ and the residual integer $k'_{j,j'} = \sum_{i=0}^{j'} 2^{i} d_{j'-i} $.

\begin{proof}[Proof of Theorem~\ref{T:NormalityLinEst}]
To keep the notation simple, we omit the dependence of the quantities $J,J_0,k$ on $n$.
Let $N+1 \le k \le 2^J-1-N-1$, so that $(J,k)$ is sufficiently far away from the boundary. Since $(J,k)\to x\in (0,1)$ this is the case if $n$ is large enough. Let $\ell = J - J_0 \to \infty$ for a base level $J_0$ which tends to $\infty$. Substituting $J$ for $j$ and $J_0$ for $j'$ in \eqref{E:DyadicRepr1} we have the following representation for the integer $k=k_n$
\begin{align}\label{E:NormalityLinEst1}
		k =  2^0 d_J + 2^1 d_{J-1} + \ldots +2^{\ell-1} d_{J_0+1} + 2^\ell k_0,
\end{align}
where $k_0 = k'_{J,J'}$ is the approximating integer on the scale $J_0$.

Using the matrices $E_N,O_N$ from equations \eqref{E:EN} and \eqref{E:ON} and the corresponding transformation matrices $(E_N)_{\dn d}, (O_N)_{\dn d}$, we obtain the representation of the wavelet
estimator w.r.t.\ the base level $J_0$
\begin{align}\begin{split}\label{E:NormalityLinEst2} 
		&\begin{pmatrix}
			\log \wh{M}_{J,k-N-1,n} \\
			\vdots \\
			\log \wh{M}_{J,k+N+1,n}
		\end{pmatrix} = \Big[ (E_N \1\{ d_J = 0\} + O_N \1\{ d_{J} = 1\} ) \cdot \ldots \\
		& \quad \ldots \cdot ( E_N \1\{ d_{J_0+1} = 0\} + O_N \1\{ d_{J_0+1} = 1\} ) \Big]_{\dn d}
		\begin{pmatrix}
			\log M_{J_0,k_0-N-1,n} \\
			\vdots \\
			\log M_{J_0,k_0+N+1,n}
		\end{pmatrix}.\end{split}	
\end{align}
In particular, for $i\in [-(N+1)..(N+1)]$, using the representation from \eqref{E:NormalityLogM1}
\begin{align}
\begin{split}\label{E:ConsistenyLinEst1}
	\log \wh{M}_{J,k + i,n} &= \sum_{v=-(N+1)}^{N+1} \alpha_{v} \log M_{J_0,k_0+v,n} \\
	&= \sum_{v=-(N+1)}^{N+1} \alpha_{v} 2^{-\ell} \sum_{u=0}^{2^\ell -1} \log M_{J, u + (k_0+v)2^\ell, n}
\end{split}
\end{align}
for certain coefficients $\alpha_v$, which sum to 1 and depend on $J_0$ and $J$. So, for each coefficient $\alpha_v$, there is a corresponding finite matrix product consisting of $E_N$ and $O_N$. Thus, the $\alpha_v$ depend on $n$, we omit this dependence in the notation in the following. Note that this implies immediately that $\sum_{v} \alpha_v^2 \ge 1/(2N+3)$ uniformly which follows from minimization problem 
$$
	\min_{w_1, \ldots, w_L,\lambda} \ \Big[\sum_v w_v^2 + \lambda(\sum_{v} w_v - 1) \Big].
$$
Consider the variance of $\log {\wh M}_{J,k,n}$. Then we find for a matrix $0\neq A\in Sym(d)$
\begin{align}\begin{split}\label{E:NormalityLinEst3}
	&Var( 2^{\ell/2} \ls A, \log {\wh M}_{J,k,n} \rs )\\
	& = \sum_{v=-(N+1)}^{N+1} \alpha_{v}^2 \sum_{u=0}^{2^\ell - 1}  2^{-\ell} \ls A, {\cal C}(2^{-J} (u+(k_0+v)2^\ell) ) A \rs.
\end{split}\end{align}	
We can use the convergence of $x^{(J_n)}$ to $x$ and continuity of the covariance operator in $x$ to obtain upper and lower bounds on \eqref{E:NormalityLinEst3}.

Once more, $	| 2^{-J} (u+(k_0+v)2^\ell)  - x |\le | k_0 2^{-J_0} - x | + 2^{-J_0} + (N+1) 2^{-(J-J_0)}$, which tends to 0 as $n\to\infty$. Hence,
\begin{align}
		&\liminf_{n\to\infty} Var( 2^{\ell/2} \ls A, \log {\wh M}_{J,k,n} \rs ) = \ls A, {\cal C}(x) A \rs \ \liminf_{n\to\infty} \ \Big(\sum_{v=-(N+1)}^{N+1}  \alpha_{v}^2 \Big)   \nonumber \\
		& \ge (2N+3)^{-1} \ls A, {\cal C}(x) A \rs > 0, \label{E:NormalityLinEst4}
\end{align}
because $A\neq 0$. With the same reasoning, an upper bound is given by
\begin{align}
		&\limsup_{n\to\infty} Var( 2^{\ell/2} \ls A, \log {\wh M}_{J,k,n} \rs )= \ls A, {\cal C}(x) A \rs  \ \limsup_{n\to\infty} \ \Big( \sum_{v=-(N+1)}^{N+1} \alpha_{v}^2 \Big) \nonumber \\
		 &\le  \| {\cal C}(x) \|_{\cal L}   \| A \|_F^2 \  \limsup_{n\to\infty} \ \Big(\sum_{v=-(N+1)d}^{(N+1)d}\alpha_{v}^2 \Big). \label{E:NormalityLinEst5}
\end{align}
So, the left-hand side of \eqref{E:NormalityLinEst4} and \eqref{E:NormalityLinEst5} agree if and only if $ \sum_{v=-(N+1)d}^{(N+1)d} \alpha_{v}^2$ converges to a number in $\R_+$ as $n\to\infty$.

If $x = 2^{-j} k$ is a dyadic rational, then the iteration in \eqref{E:NormalityLinEst2} is ultimately carried out with the matrix $E_N$ only because $J \ge J_0 > j$ if $n$ is sufficiently large and $d_J = ... = d_{J_0+1} = 0$ is the representation of $x$. In particular, in this case the limit 
\begin{align}\label{E:NormalityLinEst6}
	\lim_{n\to\infty} \sum_{v=-(N+1)}^{N+1} \alpha_{v}^2 = 	 \sum_{v=-(N+1)}^{N+1} \lim_{n\to\infty} \alpha_{v}^2 = \sum_{i=1}^{2N-1} (E_{N,\infty})^2_{N,i}
\end{align}
exists in $\R_+$ by the considerations which lead to \eqref{E:LimitMatrix}. (Note that if we use the infinite representation for dyadic rationals, we multiply with $O_N$ instead.) Consequently, in this case
\begin{align*}
		\lim_{n\to\infty} Var( 2^{\ell/2} \ls A, \log {\wh M}_{J,k,n} \rs ) = \Big(  \sum_{i=1}^{2N-1} (E_{N,\infty})^2_{N,i} \Big ) \ \ls A, {\cal C}(x) A \rs \in \R_+.
\end{align*}
This completes the proof.
 \end{proof}

\begin{proof}[Proof of Proposition~\ref{P:Cor_Thm_AN_LE}]
The result follows in a similar spirit as Theorem~\ref{T:NormalityLinEst}. Let $\delta>0$ such that the uniform bounded moments condition in  \eqref{E:MomentConditionLogM} is satisfied for $2+\delta$. The Lyapunov condition for the normalized estimator 
\begin{align}\begin{split}\label{E:NormalityLinEst7}
	& Var( \ls A, \log \wh M_{J,k,n} - \E[\log \wh M_{J,k,n} ]  \rs ) ^{-1/2}   \cdot  \ls A, \log \wh M_{J,k,n} - \E[\log \wh M_{J,k,n} ] \rs 
\end{split}\end{align}
can be established with the representation from \eqref{E:ConsistenyLinEst1} and reads as follows
\begin{align*}
    \frac{\sum_{v=-(N+1)}^{N+1} (|\alpha_{v}| \ 2^{-\ell})^{2+\delta} \sum_{u=0}^{2^\ell -1} \E[ |  \langle A, \log M_{J, u + (k_0+v)2^\ell, n} - \E[ \log M_{J, u + (k_0+v)2^\ell, n}] \rangle |^{2+\delta} ] }
    { \Big( \sum_{v=-(N+1)}^{N+1} \sum_{u=0}^{2^\ell -1}  (\alpha_{v} 2^{-\ell} )^2  \E[ |\langle A, \log M_{J, u + (k_0+v)2^\ell, n} - \E[ \log M_{J, u + (k_0+v)2^\ell, n}] \rangle |^2 ] \Big)^{(2+\delta)/2} }.
\end{align*}
In order to give an upper bound on this term, we use the uniform bounded moments condition in \eqref{E:MomentConditionLogM}, the lower bound $\sum_v \alpha_v^2 \ge 1/(2N+3)$ as well as the upper bound
\begin{align*}
    \limsup_{n\to\infty} \max_{ -(N+1)  \le v \le N+1 } |\alpha_v|
    &= \limsup_{J_0 \to\infty} \lim_{J\to\infty} \max_{ -(N+1)  \le v \le N+1 } |\alpha_v|\\
    &\leq \sup_{x\in[0,1]} \max_{i,j}~ \Phi_L(x)_{i,j}  < \infty,
\end{align*}
where $\Phi_L(x)$ is the limit matrix in \eqref{E:LimitMatrix}; notice that the entries of $\Phi_L(x)$ are uniformly 
bounded due to the H\"older continuity of the fundamental solution $\phi_L$, cf. \cite{donoho1993smooth}. 

Combining these estimates, the term in question is of order $2^{-\delta \ell/2}$. Hence, the Lyapunov condition is satisfied.
 \end{proof}

\section{Technical results on the wild bootstrap in Section~\ref{Section_ConfidenceSets}}\label{Sec:Proofs_Bootstrap}

\begin{proof}[Proof of Theorem~\ref{T:ValidityBootstrap}]
The proof is split in three parts. In the first part we study the limiting behavior of the conditional variance of the bootstrap estimator. In the second part, we verify then the Lindeberg condition. In the third part, we show the amendment. To facilitate the notation, we use the following abbreviations. Set $\ell \coloneqq J-J_0$. Let $A\in Sym(d)$ and define
\begin{align*}
	\wh Z_{u,n} =	\ls A, \ \log \wh{M}_{J^*_0,J,2^\ell k + u, n} \rs, \
	Z_{u,n} =	\ls A, \log M_{J,2^\ell k + u, n} \rs, 	Z_u = \ls A, \log M_{J,2^\ell k + u} \rs
\end{align*}
for $u\in\{0,\ldots,2^\ell - 1	\}$. We choose $\delta,\gamma>0$ such that $p=\gamma(2+\delta)$, $\gamma\ge 2 $ and such that $\delta$ is sufficiently small; see below for the admissible choices.

\textit{Part 1.} The conditional variance of $\log M_{J_0,k,n}^{boot} $ is
\begin{align*}
	Var^*( 2^{-\ell/2} \ls A, \log M_{J_0,k,n}^{boot} \rs )  &=  2^{-\ell} \sum_{u=0}^{2^\ell-1}  (\wh Z_{u,n} - Z_{u,n} )^2.
\end{align*}
We split each summand in a main term and remainder terms as follows
\begin{align}\label{E:ValidityBootstrap1}
	 \wh Z_{u,n} - Z_{u,n} = (\wh Z_{u,n} - \E[ \wh Z_{u,n} ] ) +  ( \E[ \wh Z_{u,n} ] - Z_{u} ) + (Z_u - Z_{u,n}).
\end{align}
We begin with the term $2^{-\ell} \sum_{u=0}^{2^\ell-1}  (\wh Z_{u,n} - \E[ \wh Z_{u,n} ] )^2$ on which we apply the following simple consequence of the Burkholder inequality; we state this consequence as a lemma:
\begin{Lem}\label{L:Burkholder}
Let $W_1,\ldots,W_n$ be independent, real-valued, centered random variables such that $\max_i \E[ |W_i |^q]^{1/q} \le m$, for some $q\ge 2$ and an $m\in\R_+$. Then $\E[ | \sum_i W_i|^q ] \le C_q^q \ m^q \ n^{q/2} $ for a certain $C_q \in \R_+$.
\end{Lem}
\begin{proof}[Proof of Lemma~\ref{L:Burkholder}]
Using the Burkholder inequality, $\E[ |W_i|^q] \le C_q^q \E[ (\sum_i W_i^2)^{q/2} ]$ for some $C_q\in\R_+$. As $q\ge 2$, we apply the Minkowski inequality to obtain the result.
\end{proof}
Utilizing \eqref{E:LimitMatrix}, we find with Lemma~\ref{L:Burkholder} and the moment condition on $Z_{u,n}$ from \eqref{E:MomentConditionLogM} as well as the representation of the wavelet estimator from \eqref{E:ConsistenyLinEst1} that for each $\gamma,\epsilon>0$ 
\begin{align*}
		&\p\Big( 2^{-\ell} \sum_{u=0}^{2^\ell-1}  (\wh Z_{u,n} - \E[ \wh Z_{u,n} ] )^2 \ge \epsilon \Big) \le \sum_{u=0}^{2^\ell-1} \p(  (\wh Z_{u,n} - \E[ \wh Z_{u,n} ] )^2 \ge \epsilon ) \\
		&\le \epsilon^{-\gamma} \ \sum_{u=0}^{2^\ell-1} \E[  | \wh Z_{u,n} - \E[ \wh Z_{u,n} ] |^{2\gamma} ] \le C_1  \epsilon^{-\gamma} \ 2^{(J-J_0) -(J-J_0^*)\gamma },
\end{align*}
for a constant $C_1$, which does not depend on $n$. (Note that in the first inequality, we can use this quite rough estimate because the $(\wh Z_{u,n})_u$ are dependent.)
So, provided $\delta>0$ is sufficiently small, $\sum_{n} 2^{(J-J_0) -(J-J_0^*)p/(2+\delta) } < \infty$ and we can apply the Borel-Cantelli Lemma to conclude that the partial sum
\[
	2^{-\ell} \sum_{u=0}^{2^\ell-1}  (\wh Z_{u,n} - \E[ \wh Z_{u,n} ] )^2 \to 0 \quad a.s., \quad n\to\infty.
	\]
	
Moreover, there is a constant $C_2\in\R_+$ (independent of $n$) such that 
$$
    (\E[ \wh Z_{u,n} ] - Z_u )^2 \le C_2 2^{-2N J^*_0}
    $$
uniformly in $u\in \{0,\ldots,2^\ell -1 \}$ by the approximation result from \cite[Appendix A]{chau2018}. Thus, the deterministic sum $2^{-\ell} \sum_{u=0}^{2^\ell -1} (\E[ \wh Z_{u,n} ] - Z_u )^2$ converges to 0 as $n\to \infty$.

Finally, consider $2^{-\ell} \sum_{u=0}^{2^\ell-1} (Z_{u,n}-Z_u)^2$, which has expectation
\[
	2^{-\ell} \sum_{u=0}^{2^\ell-1} \ls A, {\cal C}(2^{-J_0}k + 2^{-J} u ) A \rs.
\]
Hence, using Lemma~\ref{L:Burkholder} a second time, for each $\gamma>0$ there is a $C_3\in\R_+$ (independent of $n$) such that
\begin{align*}
	&\p\Big( 2^{-\ell} \Big| \sum_{u=0}^{2^\ell-1} (Z_{u,n}-Z_u)^2 - \E[ (Z_{u,n}-Z_u)^2 ] \Big| \ge \epsilon \Big) \\
	&\le 2^{-\ell \gamma} \epsilon^{-\gamma} \ \E\Big[ \Big| \sum_{u=0}^{2^\ell-1} (Z_{u,n}-Z_u)^2 - \E[ (Z_{u,n}-Z_u)^2 ] \Big|^\gamma  \Big] \le C_3 \epsilon^{-\gamma} 2^{-(J-J_0)\gamma/2 }.
\end{align*}
Consequently, provided $\sum_n 2^{-(J-J_0)p/(2(2+\delta)) } < \infty$, which is satisfied if we choose $\delta$ sufficiently small, we have
\[
	2^{-\ell}  \sum_{u=0}^{2^\ell-1} (Z_{u,n}-Z_u)^2 - \E[ (Z_{u,n}-Z_u)^2 ]  \to 0 \quad a.s. \quad (n\to\infty).
	\]
Thus, using the convergence result of the covariance operator from \eqref{E:NormalityLogM2}
$$
	2^{-\ell} \sum_{u=0}^{2^\ell-1} (Z_{u,n}-Z_u)^2 \to  2^{J_0} \int_{k 2^{-J_0}}^{ (k+1) 2^{-J_0}} \ls A, {\cal C}(u) A \rs \diff{u} > 0 \quad a.s.
$$
In particular, $Var^*( 2^{-\ell/2} \ls A, \log M_{J_0,k,n}^{boot} \rs ) / Var( 2^{-\ell/2} \ls A, \log  M_{J_0,k,n} \rs ) \to 1$ with probability 1 as $n\to \infty$.

Consequently, we have 
\[
	Var^*( 2^{-\ell/2} \ls A, \log \wh M_{J_0,J,k,n}^{boot} \rs ) / Var( 2^{-\ell/2} \ls A, \log \wh M_{J,k,n} \rs ) \to 1
	\]
with probability 1 as $n\to \infty$.

\textit{Part 2.} For simplicity, we verify the Lindeberg condition only for $2^{\ell/2} \log M^{boot}_{J_0,k,n} = 2^{-\ell/2} \sum_{u=0}^{2^\ell-1} \log M^{boot}_{J,u+2^\ell k,n}$; the actual verification for the wavelet estimator works in the same fashion but involves a little more notation. 

Note that $\E^*[\log M^{boot}_{J,k,n}]=\log \wh M_{J^*_0,J,k,n}$ and that $\log M^{boot}_{J, k,n}  - \log \wh M_{J^*_0,J,k,n} = \epsilon^{boot}_{J,k,n}$.  Moreover, as the conditional variance of this random variable converges $a.s.$ by the previous arguments, it is sufficient to prove the following Lindeberg condition in order to verify the conditional CLT
\begin{align*}
		L^{boot}_n(\mu) &= \sum_{u=0}^{2^\ell-1}  \E^*\Big[ | 2^{-\ell/2}  \ls A,  \epsilon^{boot}_{J,u+2^\ell k,n}  \rs |^2  \1\{| 2^{-\ell/2}  \ls A, \epsilon^{boot}_{J,u+2^\ell k,n}  \rs | > \mu \} \Big] \to 0 
\end{align*}
for each $\mu>0$ with probability 1. Let $\delta>0$, then we have
\[
	L^{boot}_n(\mu) \le \mu^{-\delta} \ \E^*[ |V_{0,0,n}|^{2+\delta}] \ 2^{-\ell (1+\delta/2 )} \sum_{u=0}^{2^\ell-1} | \ls A,  \hat \epsilon^{*}_{J,u+2^\ell k,n} \rs |^{2 + \delta}.
\]
Since $\ls A, \hat \epsilon^{*}_{J,u+2^\ell k,n} \rs = \wh Z_{u,n} - Z_{u,n}$, we rely on the same decomposition as in \eqref{E:ValidityBootstrap1}.
It is a routine to show similarly as in the first part of the proof but this time with the exponent $2+\delta$ instead of 2 that
\[
	2^{-\ell} \sum_{u=0}^{2^\ell -1} |\wh Z_{u,n} - \E[ \wh Z_{u,n} ] |^{2+ \delta } = o_{a.s.}(1) 
\]
for the choices of $\delta$ and $\gamma$. Furthermore, for the choices of $\delta$ and $\gamma$,
\[
	2^{-\ell} \sum_{u=0}^{2^\ell -1} |  \E[ \wh Z_{u,n} ] - Z_{u}  |^{2+ \delta } = o(1) \text{ and }
	2^{-\ell} \sum_{u=0}^{2^\ell -1} |Z_{u,n} - Z_u  |^{2+ \delta } = O_{a.s.}(1).
	\]

Using the monotonicity of $L^{boot}_n(\mu)$ in $\mu$, this shows then that $L^{boot}_n(\mu)\to 0$ for $n\to \infty$ for all $\mu>0$ with probability 1.

\textit{Part 3.} The convergence in the Kolmogorov distance follows from a standard argument as the limiting Gaussian distribution of $\fM_{J,k,n}$ and $\fM_{J,k,n}^{boot}$ has a continuous density on $Sym(d) \cong \R^{ (d+1)d/2}$. Let $\cal N$ be a random variable having this Gaussian distribution and let $\mu>0$. Let $\Gamma$ be a finite grid on $Sym(d)$ with minimal element $\ul x$ and maximal element $\ol x$, i.e., $\ul x\le x\le \ol x$ for all $x\in \Gamma$. (Here $x\le y$ for $x,y\in Sym(d)$ if for each position $(i,j)$, $x_{i,j} \le y_{i,j}$.) For $x\in Sym(d)$ with $\ul x\le x \le \ol x$, denote $\floor{x}$ (resp. $\ceil{x}$) the elements in $\Gamma$ which are closest to $x$ in the maximum-norm and satisfy $\floor{x}\le x$ (resp. $x\le \ceil{x}$). We choose $\Gamma$ sufficiently dense in the sense that $\p({\cal N} \le \ul x ) \le \mu$, $1 - \p({\cal N} \le \ol x) \le \mu$ and $\p( {\cal N} \le \ceil{x} ) - \p({\cal N} \le \floor{x}) \le \mu$ for $\ul x \le x \le \ol x$.

Next, choose $N\in\N$ large enough such that $|\p( \fM_{J,k,n} \le x) - \p( {\cal N} \le x) | \le \mu$ and $|\p( \fM_{J,k,n}^{boot} \le x) - \p( {\cal N} \le x) | \le \mu$ for all $x\in \Gamma$ for all $n\ge N$. 

Then $|\p( \fM_{J,k,n} \le x) - \p( \fM_{J,k,n}^{boot} \le x)| \le 4\mu$ for all $x\in Sym(d)$ for all $n\ge N$.
\end{proof}


\section{The log-Euclidean Metric}\label{AppendixA}


For the readers ease and clarity in notation, we summarize the results of \cite{arsigny2007geometric} necessary for this work.  
We also add the formula for parallel transport under the log-Euclidean metric. For generel treatments of modern differential
and riemannian geometry we refer to \cite{Lee2013SmoothManifolds}, \cite{Lee2013RiemannianManifolds} as well as 
\cite{Tu2013DifferentialGeometry}.

\subsection{Preliminaries}
\label{Subsec_PreliminariesLogEuclidean}

\begin{Thrm}[cf. \cite{arsigny2007geometric}, Theorem 2.2]
The matrix exponential $\exp\colon M(d) \to GL(d)$ is a $C^{\infty}$-mapping and its differential $d\exp\colon TM(d) \cong M(d) \to M(d) \cong T GL(d)$ is given by 
\begin{align*}
	d_A\exp\colon &T_A M(d) \cong M(d) \to M(d) \cong T_{\exp(A)} GL(d),\\
	& \quad B \mapsto \sum_{k=1}^{\infty} \frac{1}{k!} \sum_{l=0}^{k-1} A^lBA^{k-l-1}.
\end{align*}
\end{Thrm}
\begin{proof}
Smoothness of the matrix exponential follows simply from the uniform absolute convergence of the series. The differential of $\exp$ is given by
	\begin{align*}
	d_A\exp(B) 
	&= \frac{d}{dt} \left.\exp(A+tB)\right|_{t=0}
	= \sum_{k=0}^{\infty} \frac{1}{k!} \frac{d}{dt} \left. (A+tB)^k \right|_{t=0}\\
	&= \sum_{k=0}^{\infty} \frac{1}{k!} \left. \sum_{l=0}^{k-1} (A+tB)^l B (A+tB)^{k-l-1} \right|_{t=0}\\
	&= \sum_{k=1}^{\infty} \frac{1}{k!} \sum_{l=0}^{k-1} A^lBA^{k-l-1}.
	\end{align*}
\end{proof}

\begin{Cor}[cf. \cite{arsigny2007geometric}, Cor. 2.3]
In particular $d_0 \exp = \Id\colon M(d) \to M(d)$ and
\begin{align*}
\trace(d_A\exp(B)) = \trace(\exp(A)B),\quad A,B \in M(d).
\end{align*}
\end{Cor}
\begin{proof}
A direct computation shows
\begin{align*}
	\trace(d_A\exp(B)) &= \sum_{k=1}^{\infty} \frac{1}{k!} \sum_{l=0}^{k-1} \trace(A^lBA^{k-l-1}) 
	= \sum_{k=1}^{\infty} \frac{1}{k!} \sum_{l=0}^{k-1} \trace(A^{k-1}B)\\
	&= \trace\left( \sum_{k=1}^{\infty} \frac{A^{k-1}}{(k-1)!} B\right) = \trace(\exp(A)B).
\end{align*}
For the other assertion notice that $T_{\Id}GL(d) = M(d)$ (Lie algebra of the Lie group $GL(d)$).
\end{proof}

\begin{Def}[matrix logarithm]
A matrix $B\in M(d)$ is called a logarithm of a matrix $A\in GL(d)$ if $\exp(B) = A$. 
\end{Def}

\begin{Rem}[cf. \cite{arsigny2007geometric}]
Since the mapping $\exp \colon M(d) \to GL(d)$ is \emph{not surjective} the logarithm of a matrix may not exist in general. However, if $A\in GL(d)$ has no (complex)
eigenvalues
on the (closed) negative real line, then $A$ has a \emph{unique} real logarithm whose (complex) eigenvalues have an imaginary part in $(-\pi,\pi)$. This particular 
logarithm is called
\emph{principal logarithm} and will be denoted $\log(A)$ whenever it is defined. Especially $\log(A)$ is defined for any $A\in Sym^+(d)$ and is symmetric.
\end{Rem}

\subsection{Log-Euclidean Metric(es) on \texorpdfstring{$Sym^+(d)$}{Sym+(d)}}
\label{Subsec_LogEuclideanMetrices}

The following theorem summarizes \cite{arsigny2007geometric} Theorem 2.6, Proposition 2.7 and Theorem 2.8.

\begin{Thrm}
The mappings $\exp\colon Sym(d) \to Sym^+(d)$ and $\log \colon Sym^+(d) \to Sym(d)$ are $C^{\infty}$ and one-to-one, i.e., the spaces are diffeomorphic and 
$\exp^{-1} = \log$. Also 
\begin{align*}
	d_A\exp \colon T_A Sym(d) \cong Sym(d) \to Sym(d) \cong T_{exp(A)} Sym^+(d)
\end{align*}
is invertible for all $A \in Sym(d)$.
Topologically $Sym^+(d)$ is an open convex half-cone of $Sym(d)$ and therefore a submanifold of $Sym(d)$.  
\end{Thrm}

The idea for the construction of log-Euclidean metrices is to use the matrix exponential to transport the additive group structure of
$Sym(d)$ to $Sym^+(d)$. To this end define the \emph{logarithmic product}
$\odot\colon Sym^+(d) \times Sym^+(d) \to Sym^+(d)$ by
\begin{align*}
	S_1 \odot S_2 := \exp(\log(S_1) + \log(S_2)).
\end{align*}

\begin{Thrm}\label{ThrmLieGroup}
$(Sym^+(d), \odot)$ is a commutative Lie group. The neutral element is the identity matrix and the inverse element is simply the matrix inverse. Furthermore the 
matrix exponential
\begin{align*}
	\exp \colon (Sym(d), +) \to (Sym^+(d), \odot)
\end{align*}
is a Lie group isomorphism.\\

In particular one-parameter subgroups of $(Sym^+(d),\odot)$ are of the form $\exp(tV)_t$, where $(tV)_t$, $V\in Sym(d)$ are simply the
one-parameter subgroups of $(Sym(d), +)$. Also the Lie group exponential 
\begin{align*}
	\exp\colon T_{\Id} Sym^+(d) = Sym(d) \to Sym^+(d) 
\end{align*}
is just the ordinary matrix exponential.
\end{Thrm}
\begin{proof}
\cite{arsigny2007geometric} Proposition 3.2, Theorem 3.3 and Proposition 3.4 
\end{proof}

Now any metric (inner product) $\ls\cdot,\cdot\rs_{\Id}$ on $Sym(d) = T_{\Id}Sym^+(d)$ can be extended to a \emph{Riemannian metric} $g$ on $(Sym^+(d), \odot)$ by
\begin{align}\label{metricextension}
	g_S(U,V) = \ls U,V\rs_{S} :=& \ls (d_{\Id}l_S)^{-1}U, (d_{\Id}l_S)^{-1}V \rs_{\Id}, \\
	=& \ls (d_{\Id}r_S)^{-1}U, (d_{\Id}r_S)^{-1}V \rs_{\Id},\nonumber
\end{align}
because $U,V \in T_SSym^+(d) \cong Sym(d),~~S\in Sym^+(d)$ and where 
\begin{align*}
	l_S&\colon Sym^+(d) \to Sym^+(d),~~ W \mapsto S\odot W \quad \text{(\emph{left translation})}\\
	d_{\Id}l_S&\colon T_{\Id}Sym^+(d) \to T_S Sym^+(d).
\end{align*}
Since the logarithmic product $\odot$ is commutative, we have $l_S = r_S$, where 
\begin{align*}
	r_S&\colon Sym^+(d) \to Sym^+(d),~~ W \mapsto W\odot S
\end{align*}
denotes the \emph{right translation}. Therefore we have
\begin{Cor}[cf. \cite{arsigny2007geometric}, Cor. 3.7 and Cor. 3.10]
Any Riemannian metric obtained by (\ref{metricextension}) is a bi-invariant metric on $(Sym^+(d), \odot)$, i.e.
\begin{align*}
	\ls U,V \rs_{\Id} &= \ls d_{\Id}l_S(U), d_{\Id}l_S(V) \rs_S  &&\text{(left-invariant)}\\
					 &\big(= \ls d_{\Id}r_S(U), d_{\Id}r_S(V) \rs_S\big).  &&\text{(right-invariant)}
\end{align*}
Also endowed with such a metric $Sym^+(d)$ becomes a flat Riemannian manifold, i.e., the curvature tensor (resp. the sectional curvature) is null. 
\end{Cor}

\begin{Def}[log-Euclidean metric, cf. \cite{arsigny2007geometric}, Def. 3.8]
Any bi-invariant metric on $(Sym^+(d), \odot)$ is called log-Euclidean metric.
\end{Def}

\subsection{Geometry of \texorpdfstring{$Sym^+(d)$}{Sym+(d)} under log-Euclidean Metric(es)}
Denote by $g$ a bi-invariant metric on $(Sym^+(d), \odot)$ obtained by (\ref{metricextension}). Recall that 
	\begin{align*}
	\frac{d}{dt} \left.\exp(S+tV) \right|_{t=0} = d_S\exp(V) \in T_{\exp(S)}Sym^+(d),\quad V\in T_S Sym(d)
	\end{align*}
Define $\gamma_S(\cdot, L)\colon [0,1] \to Sym^+(d)$ for $S\in Sym^+(d)$ und $L \in T_SSym^+(d)$ by
\begin{align*}
	\gamma_S(t;L) = \exp\big(\log(S) + t (d_{\log(S)}\exp)^{-1}L\big),
\end{align*}
then $\gamma_S(0;L) = S$ and $\dot{\gamma_S}(0;L) = L$. Since $\gamma_S(\cdot;L)$ is a one-parameter subgroup of $(Sym^+(d), \odot)$ (cf. Theorem \ref{ThrmLieGroup})
by Theorem 3.6 in \cite{arsigny2007geometric} $\gamma_S(\cdot;L)$ is a geodesic (with resp. to $g$) starting in $S$ in direction $L$. Therefore the \emph{exponential
map} $\Exp_S \colon T_SSym^+(d) \to Sym^+(d)$ is given by 
\begin{align*}
	\Exp_S(L) := \gamma_S(1;L) = \exp\big(\log(S) + (d_{\log(S)}\exp)^{-1}L\big).
\end{align*}
Now consider $\log\circ \exp = \Id\colon Sym^+(d) \to Sym^+(d)$, then
\begin{align*}
	&~~d_S(\log\circ \exp)(U) = d_{\exp(S)}\log \circ~ d_S\exp(U) \\
	&\qquad\qquad\qquad\quad \ \ = d_S\Id(U) = \frac{d}{dt}\left.(S+tU)\right|_{t=0} = U\quad \forall U \in Sym^+(d)\\
	&\Leftrightarrow~~ d_{\exp(S)}\log \circ~ d_S\exp = \Id\\
	&\Leftrightarrow~~ d_{\exp(S)}\log = (d_S\exp)^{-1}\\
	&\Leftrightarrow~~ d_S\log = (d_{\log(S)}\exp)^{-1},
\end{align*}
which yields
\begin{align}\label{ExpMap}
	\Exp_S(L) = \exp\big(\log(S) + d_S\log (L)\big).
\end{align}
Next put $\Exp_{S_1}(L) = S_2$ for $S_1,S_2 \in Sym^+(d)$.
Solving (\ref{ExpMap}) for $L\in T_{S_1}Sym^+(d)$ gives
\begin{align*}
	\log(S_1) + d_{S_1}\log(L) = \log(S_2) 
	&\quad\Leftrightarrow\quad d_{S_1}\log(L) = \log(S_2)- \log(S_1)\\
	&\quad\Leftrightarrow\quad L = d_{\log(S_1)}\exp\big( \log(S_2) - \log(S_1) \big).
\end{align*} 
Therefore the \emph{logarithmic map} $\Log_{S_1} \colon Sym^+(d) \to T_{S_1}Sym^+(d)$ is given by
\begin{align*}
	\Log_{S_1}(S_2) = d_{\log(S_1)}\exp\big( \log(S_2) - \log(S_1) \big).
\end{align*}
Furthermore for $\gamma(t) = \exp(tU)$, $U \in Sym(d)$ we have $\gamma(0) = \Id$ and $\dot{\gamma}(0) = U$, thus 
\begin{align*}
	&d_{\Id}l_S(U) = \frac{d}{dt} \left.l_S\circ \gamma(t)\right|_{t=0} \\
	&= \frac{d}{dt} \left. S\odot \exp(tU)\right|_{t=0}
	= \frac{d}{dt} \left. \exp(\log(S) + tU)\right|_{t=0}
	= d_{\log(S)}\exp(U)\\
	\Leftrightarrow&~~(d_{\Id}l_S)^{-1} = (d_{\log(S)}\exp)^{-1} = d_S\log.
\end{align*}
The \emph{log-Euclidean metric} $g$ can consequently be explicitly written as 
\begin{align}\label{LogEuclMetric}
	g_S(U,V) = \ls d_S\log(U), d_S\log(V)\rs_{\Id}.
\end{align} 
Accordingly for the \emph{distance of two points} $S_1,S_2 \in Sym^+(d)$ we obtain
\begin{align*}
	d^2(S_1,S_2) 
	&= g_{S_1}^2(\Log_{S_1}(S_2), \Log_{S_1}(S_2))~ (= \| \Log_{S_1}(S_2) \|_{S_1}^2)\\
	&= \ls d_{S_1}\log(\Log_{S_1}(S_2)), d_{S_1}\log(\Log_{S_1}(S_2))\rs_{\Id}\\
	&= \| \log(S_2) - \log(S_1) \|_{\Id}^2.
\end{align*}

To conclude the discussion we note that Theorem 3.6 in \cite{arsigny2007geometric} together with Theorem \ref{ThrmLieGroup} immediately imply that the unique 
geodesic from $S_1$ to $S_2$ in $(Sym^+(d), \odot)$ is 
\begin{align*}
	\gamma(t;S_1,S_2) = \exp((1-t)\log(S_1) + t\log(S_2)), \quad t\in[0,1].
\end{align*}

Also (\ref{LogEuclMetric}) shows that $\log: (Sym^+(d), g) \to (Sym(d), \ls\cdot,\cdot\rs_{\Id})$ is an \emph{isometry} (and likewise 
$\exp\colon(Sym(d), \ls\cdot,\cdot\rs_{\Id}) \to (Sym^+(d), g)$), so the following diagram commutes  
\begin{align*}
    \begin{array}{ccc}
    T_{S_1}Sym^+(d) & \xrightarrow{\Gamma_{S_1}^{S_2}} & T_{S_2}Sym^+(d) \\
    &&\\
    \Big\downarrow d_{S_1}\log & & \Big\downarrow d_{S_2}\log\\
    &&\\
    T_{\log(S_1)}Sym(d) & \xrightarrow{\widetilde\Gamma_{\log(S_1)}^{\log(S_2)}} & T_{\log(S_2)}Sym(d)
    \end{array}    
\end{align*}
where $\Gamma_{S_1}^{S_2}$ denotes the \emph{parallel transport} (in $(Sym^+(d),g)$) along the geodesic $\gamma(\cdot;S_1,S_2)$. The corresponding parallel 
transport $\widetilde\Gamma_{\log(S_1)}^{\log(S_2)}$ in $Sym(d)$ (along $\log \circ \gamma$) is just the identity map since $Sym(d)$ is a vector space. Therefore 
for $U\in T_{S_1}Sym^+(d)$ we have
\begin{align*}
	\Gamma_{S_1}^{S_2}(U)
	&= (d_{S_2}\log)^{-1} \circ d_{S_1}\log(U)\\ 
	&= d_{\log(S_2)}\exp\big( d_{S_1}\log(U) \big).
\end{align*}
In particular the parallel transport of $U \in T_SSym^+(d)$ to the tangent space at the identity, $T_{\Id}Sym^+(d)= Sym(d)$, is 
\begin{align*}
	\Gamma_{S}^{\Id}(U) = d_S\log(U)
\end{align*}
since $d_{\log(\Id)}\exp = d_0\exp = \Id$, cf. \cite{yuan2012local} Equation (14).\\

Finally, by choosing $\ls U, V \rs_{\Id} = \ls U, F \rs_F = \trace(UV)$ the corresponding log-Euclidean metric is 
invariant under orthogonal similarity transformations.

\begin{Lem}\label{Lem:Invariance}
For $O \in GL(d)$ orthogonal and $S_1,S_2 \in Sym^+(d)$ we have 
\begin{align*}
    d(OS_1O^T, OS_2O^T) = d(S_1,S_2).
\end{align*}
\end{Lem}
\begin{proof}
Simply observe that $\log(ASA^{-1}) = A\log(S)A^{-1}$ for any non-singular matrix $A \in M(d)$ and $S \in Sym^+(d)$. Hence 
\begin{align*}
    d(OS_1O^T, OS_2O^T) &= \| O (\log S_1- \log S_2) O^T \|_F\\
    &= \| \log S_1- \log S_2 \|_F = d(S_1 , S_2).
\end{align*}
\end{proof}

\end{document}